\documentclass[10pt]{article}

\usepackage[margin=1in]{geometry}
\usepackage{amsmath,array,bm}
\usepackage{amsmath,amsfonts}
\usepackage{amsthm}
\usepackage{amssymb,latexsym}
\usepackage{algorithm}
\usepackage{subcaption}
\usepackage{algorithmicx}
\usepackage[]{algpseudocode}
\usepackage[dvipsnames]{xcolor}
\usepackage{graphicx}
\usepackage{caption}
\usepackage{enumitem}
\usepackage{float}
\usepackage{thmtools}
\usepackage{hyperref}
\usepackage{mathtools}
\usepackage{cleveref}
\usepackage{bussproofs}
\usepackage[square,sort,comma,numbers]{natbib}
\usepackage{tikz}
\usepackage{xspace}
\usepackage{xcolor}
\usepackage{tipa}
\usepackage{mdframed}
\usepackage[export]{adjustbox}[2011/08/13]
\usepackage{multirow}
\usepackage{hhline}
\usepackage{todonotes}
 \usepackage{tcolorbox}

\newtheorem{theorem}{Theorem}
\newtheorem{claim}{Claim}

\theoremstyle{definition}
\newtheorem{definition}{Definition}
\newtheorem{observation}{Observation}

\newtheorem{lemma}{Lemma}

\newtheorem{corollary}{Corollary}

\crefname{invar}{invariant}{invariants}
\crefname{ineq}{inequality}{inequalities}
\crefname{constr}{constraint}{constraints}
\crefname{tbl}{table}{tables}
\crefname{lem}{lemma}{lemmata}
\crefname{lemma}{lemma}{lemmata}
\crefname{cond}{condition}{conditions}

\title{ETH Tight Algorithms for Geometric Intersection Graphs: Now in Polynomial Space}
\author{Fedor V. Fomin \footnote{University of Bergen, Norway. \texttt{Fedor.Fomin@uib.no} } \and Petr A. Golovach \footnote{University of Bergen, Norway. \texttt{Petr.Golovach@uib.no}} \and Tanmay Inamdar\footnote{University of Bergen, Norway. \texttt{Tanmay.Inamdar@uib.no}} \and Saket Saurabh\footnote{The Institute of Mathematical Sciences, HBNI, Chennai, India, and University of Bergen, Norway. \texttt{saket@imsc.res.in}} \footnote{The research leading to these results has received funding from the Research Council of Norway via the project ```MULTIVAL" (grant no. 263317) and the European Research Council (ERC) via grant LOPPRE, reference 819416.}}
\date{}
\newcommand{\Oh}{\mathcal{O}}
\newcommand{\real}{\mathbb{R}}
\renewcommand{\natural}{\mathbb{N}}
\renewcommand{\P}{\mathcal{P}}
\newcommand{\Gp}{G_\mathcal{P}}
\newcommand{\I}{\mathcal{I}}

\renewcommand{\k}{\kappa}
\renewcommand{\S}{\mathcal{S}}
\newcommand{\R}{\mathcal{R}}
\newcommand{\Q}{\mathcal{Q}}
\newcommand{\kpart}{$\kappa$-partition\xspace}
\ifdefined\C
\renewcommand{\C}{\mathcal{C}}
\else
\newcommand{\C}{\mathcal{C}}
\fi

\newcommand{\lr}[1]{\left( #1\right)}
\newcommand{\LR}[1]{\left\{ #1\right\}}

\newcommand{\dd}{{1-1/d}}
\newcommand{\ndd}{\Oh(n^\dd)}
\newcommand{\tndd}{2^{\ndd}}
\newcommand{\tw}{\texttt{tw}}
\newcommand{\cnc}{Cut\&Count\xspace}
\newcommand{\zz}{\mathbf{0}}
\newcommand{\oo}{\mathbf{1}}

\newcommand{\ts}{\widetilde{S}}

\newcommand{\wtd}{\texttt{wtd}}
\newcommand{\td}{\texttt{td}}

\newcommand{\lrcut}{(X_L, X_R)}
\newcommand{\states}{\texttt{states}}
\newcommand{\F}{\mathcal{F}}
\newcommand{\tail}{\texttt{tail}}
\newcommand{\tree}{\texttt{tree}}
\newcommand{\broom}{\texttt{broom}}
\newcommand{\child}{\texttt{child}}
\newcommand{\cc}{\texttt{cc}}
\newcommand{\w}{\mathbf{w}}
\newcommand{\integer}{\mathbb{Z}}
\newcommand{\bip}{\textbf{bip}}

\newif\iflong
\longtrue   

\begin{document}
	\maketitle

\begin{abstract}
De Berg et al. in [SICOMP 2020] gave an algorithmic framework for subexponential algorithms on geometric graphs with tight (up to ETH) running times. This framework is based on dynamic programming on graphs of weighted treewidth resulting in algorithms that use super-polynomial space. We introduce the notion of weighted treedepth and use it to refine the framework of de Berg et al. for obtaining polynomial space (with tight running times)  on geometric graphs. As a result, we prove that  for any fixed dimension $d\ge 2$  on intersection graphs of similarly-sized fat objects many well-known graph problems including \textsc{Independent Set}, $r$-\textsc{Dominating Set} for constant $r$, \textsc{Cycle Cover}, \textsc{Hamiltonian Cycle}, \textsc{Hamiltonian Path}, \textsc{Steiner Tree}, \textsc{Connected Vertex Cover}, \textsc{Feedback Vertex Set}, and \textsc{(Connected) Odd Cycle Transversal}  are solvable in  time $2^{O(n^{1-1/d})}$ and within polynomial space.  
\end{abstract}

\section{Introduction}\label{fv:secintro} Most of the fundamental NP-complete problems on graphs like \textsc{Independent Set}, \textsc{Feedback Vertex Set}, or \textsc{Hamiltonian Cycle} do not admit algorithms of running times $2^{o(n)}$ on general graphs unless the Exponential Time Hypothesis (ETH) fails. However, on planar graphs, $H$-minor-free graphs, and several classes of geometric graphs, such problems admit \emph{subexponential} time algorithms. There are several general frameworks for obtaining subexponential algorithms \cite{BergBKMZ20,DemaineFHT05jacm,FominLPSZ19}. The majority of these frameworks utilize dynamic programming algorithms over graphs of bounded treewidth. Consequently, the subexponential algorithms derived within these frameworks use prohibitively large (exponential) space.

Recently, algorithms on graphs of bounded treedepth attracted significant attention  \cite{FurerY17,HegerfeldK20,NederlofPSW20}.  The advantage of these algorithms over dynamic programming used for treewidth is that they use polynomial space. Our work is motivated by the following natural question 
   \begin{tcolorbox}[colback=green!5!white,colframe=blue!40!black]
Could the treedepth   find applications in the design of (polynomial space) subexponential algorithms?
  \end{tcolorbox} 
The problem is that the treedepth of a graph could be significantly larger than its treewidth. For example, the treewidth of an $n$-vertex path is one, while the treedepth is of order $\log{n}$. It creates problems in using treedepth in frameworks like bidimensionality that strongly exploit the existence of large grid minors in graphs of large treewidth. Despite that, we show the usefulness of treedepth for obtaining polynomial space subexponential algorithms on intersection graphs of some geometrical objects.  

  In~\cite{BergBKMZ20}, de Berg et al. developed a generic framework facilitating the construction of subexponential algorithms on large classes of geometric graphs. By applying their framework on intersection graphs of similarly-sized fat objects in dimension $d\geq 2$, de Berg et al. obtained algorithms with running time $\tndd$ for many well-known graph problems, including \textsc{Independent Set}, $r$-\textsc{Dominating Set} for constant $r$, \textsc{Hamiltonian Cycle}, \textsc{Hamiltonian Path},  \textsc{Feedback Vertex Set},   \textsc{Connected Dominating Set}, and \textsc{Steiner Tree}. 

The primary tool introduced by de Berg et al. is the weighted treewidth. They show that solving many optimization problems on   intersection graph of $n$ similarly-sized fat objects can be reduced on solving these problems on graphs of 
weighted treewidth of order  $\Oh(n^{\dd})$. Combined with single-exponential algorithms on graphs of bounded weighted treewidth, this yields subexponential algorithms for several problems.  

 The running times $\tndd$  are tight---de Berg et al. accompanied their algorithmic upper bounds with matching conditional complexity (under  ETH) bounds. However, as most of the treewidth-based algorithms, the algorithms of Berg et al. are dynamic programming over tree decompositions. As a result,  they require super-polynomial space. Thus a concrete question here is \emph{whether running times $\tndd$ could be achieved using polynomial space.}

 We answer this question affirmatively by developing polynomial space algorithms that in time $\tndd$ solve all problems on intersection graphs of similarly-sized fat objects from the paper of de Berg et al. except for \textsc{Connected Dominating Set}. The primary tool in our work is the \emph{weighted treedepth}. To the best of our knowledge, this notion is new. 
 
The \cnc technique was introduced by Cygan et al.\ \cite{cygan2011solving}, who gave the first single-exponential (randomized) algorithms parameterized by the treewidth for many problems using this technique. We note that at the heart of these algorithms is a dynamic programming over the tree decomposition, and thus require exponential space. However, unweighted treedepth was recently used by several authors in the design of parameterized algorithms using polynomial space 
\cite{FurerY17,HegerfeldK20,NederlofPSW20}. Some of these works adapt the \cnc technique for the treedepth decomposition.

Our main insight is that in the framework of de Berg et al. \cite{BergBKMZ20} for most of the problems the weighted treedepth can replace the weighted treewidth.  Pipelined with branching algorithms over graphs weighted treedepth, this new insight brings us to many tight (up to ETH) polynomial space algorithms on geometric graphs.

\medskip\noindent\textbf{Our results.} 
To explain our strategy of ``replacing'' the weighted treewidth with the weighted treedepth,  we need 
to provide an overview of the framework of de Berg et al. \cite{BergBKMZ20}. It has two main ingredients. First, for an intersection graph of $n$ similarly-sized fat objects (we postpone technical definitions to the next section), we construct an auxiliary weighted graph $\Gp$. (Roughly speaking, to create $\Gp$, we contract some cliques of  $G$ and assign weights to the new vertices.)   Then the combinatorial theorem of de Berg et al. states that the weighted treewidth of $\Gp$ is $\Oh(n^{\dd})$. Second, to solve problems on $G$ in time $\tndd$, one uses a tree decomposition of  $\Gp$. This part is problem-dependent and, for some problems, could be pretty non-trivial.

To plug in the treedepth into this framework, we first prove that the weighted treedepth of $\Gp$ is $\Oh(n^{\dd})$. Moreover, we give an algorithm computing a  treedepth decomposition in time  $\tndd$ and polynomial space. For \textsc{Independent Set}, a simple branching algorithm over the treedepth decomposition can solve the problem in time $\tndd$ and polynomial space. We also get a similar time and space bounds for \textsc{Dominating Set}, and more generally, $r$-\textsc{Dominating Set} for constant $r$; however, we need to use a slightly different kind of recursive algorithm.

Next, we consider connectivity problems like \textsc{Steiner Tree, Connected Vertex Cover}, \textsc{Feedback Vertex Set}, and \textsc{(Connected) Odd Cycle Transversal}. For these problems, we are able to adapt the single exponential FPT algorithms parameterized by (unweighted) treedepth given by Hegerfeld and Kratsch \cite{HegerfeldK20}, into the framework of weighted treedepth decomposition. Thus, we get $\tndd$ time, polynomial space algorithms for these problems.

Finally, we consider \textsc{Cycle Cover}, which is a generalization of \textsc{Hamiltonian Cycle}. Here, we are able to ``compress'' the given graph into a new graph, such that the (unweighted) treedepth of the new graph is $\ndd$. We can also compute the corresponding treedepth decomposition in $\tndd$ time, and polynomial space. Then, we can use a result by Nederlof et al.\ \cite{NederlofPSW20} as a black box, which is \cnc based a single exponential FPT algorithms parameterized by treedepth, that use polynomial space. Thus, we get $\tndd$ time, polynomial space algorithms for \textsc{Cycle Cover, Hamiltonian Cycle}, and also for \textsc{Hamiltonian Path}.

We note that the results in the previous two paragraphs are based on the \cnc technique, and are randomized. We also note that all of our algorithms, except for \textsc{Cycle Cover} and related problems can work even without the geometric representation of the similarly-sized fat objects. For \textsc{Cycle Cover} and related problems, however, we require the geometric representation. This is in line with similar requirements for these problems from \cite{BergBKMZ20}.

\paragraph{Organization.} In Section \ref{fv:sec:prelim}, we define some of the basic concepts including the weighted treedepth, and then prove our main result about the same. In Section \ref{fv:sec:recursive}, we discuss algorithms for \textsc{Indpendent Set}, and \textsc{$r$-Dominating Set}. Then, in Section \ref{fv:sec:cnc}, we discuss the \cnc technique, and its application for \textsc{Steiner Tree, Connected Vertex Cover, Feedback Vertex Set}, and \textsc{(Connected) Odd Cycle Transversal}. Finally, in Section \ref{fv:sec:cyclecover}, we describe our algorithm for \textsc{Cycle Cover} and related problems.


\section{Geometric Graphs and Weighted Treedepth} \label{fv:sec:prelim}

In this section we define the weighted treedepth, prove a combinatorial bound on the treedepth of certain geometric graphs and provide a generic algorithm and provide an abstract theorem  modeling at a high level our subexponential time and polynomial space algorithms. But first, we need some definitions.

\medskip\noindent\textbf{Graphs.}
We consider only undirected simple graphs and use the standard graph theoretic terminology; we refer to the book of Diestel~\cite{Diestel12} for basic notions. We write $|G|$ to denote $|V(G)|$, and 
throughout the paper we use $n$ for the number of vertices if it does not create confusion. For a set of vertices $S\subseteq V(G)$, we denote by $G[S]$ the subgraph of $G$ induced by the vertices from $S$ and write $G-S$
to denote the graph obtained by deleting the vertices of $S$. For a vertex $v$, $N_G(v)$ denotes the \emph{open neighborhood} of $v$, that is, the set of vertices adjacent to $v$, and $N_G[v]=\{v\}\cup N_G(v)$ is the \emph{closed neighborhood}. For a vertex $v$, $d_G(v)=|N_G(v)|$ denotes the \emph{degree} of $v$.  We may omit subscripts if it does not create confusion. 
For two distinct vertices $u$ and $v$ of a graph $G$, a set $S\subseteq V(G)$ is a \emph{$(u,v)$-separator} if $G-S$ has no $(u,v)$-path and $S$ is a \emph{separator} if $S$ is a $(u,v)$-separator for some vertices $u$ and $v$.
A pair of vertex subsets $(A, B)$ is called a separation if $A \cup B = V(G)$, and there are no edges between $A\setminus B$ and $B \setminus A$, that is, $S=A\cap B$ is a $(u,v)$-separator for $u\in A\setminus B$ and $v\in B\setminus A$. We say that a subset $S \subseteq V(G)$ is an $\alpha$-balanced separator for a constant $\alpha \in (0, 1)$ if there exists a separation $(A, B)$ such that $A \cap B = S$, and $\max\LR{|A|, |B|} \le \alpha n$.

\medskip\noindent\textbf{$\k$-partition.} Let $\P = \{V_1, V_2, \ldots, V_t\}$ be a partition of $V(G)$ for some $t \ge 1$, such that any $V_i \in \P$ satisfies the following properties: (1) $G[V_i]$ is connected, and (2) $V_i$ is a union of at most $\k$ cliques in $G$ (not necessarily disjoint). Then, we say that $\P$ is a \kpart of $G$. Furthermore, given a \kpart $\P = \{V_1, V_2, \ldots, V_t\}$ of $G$, we define the graph $G_\P$, the graph induced by $\P$,  as the undirected graph obtained by contracting each $V_i$ to a vertex, and removing self-loops and multiple edges. 

\medskip\noindent\textbf{Treedepth and Weighted Treedepth.}
We introduce  \emph{weighted treedepth} of a graph as a generalization of the well-known notion of treedepth (see e.g. the book of Nesetril and de Mendez~\cite{NesetrilM12}). There are different ways to define treedepth but it is convenient for us to deal with the definition via \emph{treedepth decompositions} or \emph{elimination forests}. 
We say that a forest $F$  supplied with one selected node (it is convenient for us to use the term ``node'' instead of ``vertex'' in such a forest) in each connected component, called a \emph{root}, a \emph{rooted forest}. The choice of roots defines the natural parent--child relation on the nodes of a rooted forest. Let $G$ be a graph and let $\omega\colon V(G)\rightarrow\mathbb{R}$ be a weight function. A \emph{treedepth decomposition} of $G$ is a pair $(F,\varphi)$, where $F$ is a rooted forest and $\varphi\colon V(F)\rightarrow V(G)$ is a bijective mapping such that for every edge $uv\in E(G)$, either $\varphi^{-1}(u)$ is an ancestor of $\varphi^{-1}(v)$ in $F$ or $\varphi^{-1}(v)$ is an ancestor of $\varphi^{-1}(u)$. Then the \emph{depth} of the decomposition is the depth of $F$, that is, the maximum number of nodes in a path from a root to a leaf. The \emph{treedepth} of $G$, denoted $\td(G)$, is  the minimum depth of a treedepth decomposition of $G$. We define the \emph{weighted depth} of a treedepth decomposition as the maximum $\sum_{v\in V(P)}\omega(\varphi(v))$ taken over all paths $P$ between roots and leaves. Respectively, the \emph{weighted treedepth} $\wtd(G)$ is the minimum weighted depth of a treedepth decomposition. For our applications, we assume without loss of generality  that $G$ is connected, which implies that the forest $F$ in a (weighted) treedepth decomposition is actually a tree. 

\medskip\noindent\textbf{Weighted Treewidth.} We assume basic familiarity with the notion of treewidth and tree decomposition of a graph -- see a text such as \cite{cygan2015parameterized}, for example.  Similar to the previous paragraph, de Berg et al. \cite{BergBKMZ20} define the \emph{weighted treewidth} of a graph. Given an undirected graph $G = (V, E)$ with weights $\omega\colon V(G)\rightarrow\mathbb{R}$, the weighted width of a tree decomposition $(T, \beta)$, is defined to be the maximum over bags, the sum of the weights of vertices in the bag. The weighted treewidth of a graph is the minimum weighted width over all tree decompositions of the graph.

It is useful to observe that we consider treedepth and tree decompositions of the graphs $\Gp$ constructed for graphs $G$ with given \kpart  $\P = \{V_1, V_2, \ldots, V_t\}$. Then the treedeph decomposition of $\Gp$ can be seen as a pair $(F,\varphi)$, where $F$ is a rooted forest and $\varphi$ is a bijective mapping of $V(F)$ to $\P$. Similarly, in a tree decomposition $(T, \beta)$ of $\Gp$, corresponding to every node $t \in V(T)$, the bag $\beta(t)$ is a subset of $\P$. Finally, we observe that the results of \cite{BergBKMZ20} regarding weighted treewidth---thus our results for weighted treedepth---hold for any weight functions $\omega: \P \to \real^+$, provided that $\omega(t) = \Oh(t^{\dd-\epsilon})$, for any $\epsilon > 0$. However, as in \cite{BergBKMZ20}, it suffices to fix $\omega(V_i) \coloneqq \log(1+|V_i|)$ for our applications. Thus, the weight function is assumed to be the aforementioned one, unless explicitly specified otherwise. For the simplicity of notation, we use the shorthand $\omega(u_i) \coloneqq \omega(\varphi(u_i))$ for any $u_i \in V(F)$, and $\omega(S) \coloneqq \sum_{u_i \in S} \omega(u_i)$ for any subset $S \subseteq V(F)$. 

\medskip\noindent\textbf{Geometric Definitions.} Given a set $F$ of objects in $\real^d$, we define the corresponding intersection graph $G[F] = (V, E)$, where there is a bijection between an object in $F$ and $V(G)$, and $uv \in E(G)$ iff the corresponding objects in $F$ have a non-empty intersection. It is sometimes convenient to erase the distinction between  $F$ with $V(G)$, and to say that each vertex is a geometric object from $F$.

We consider the geometric intersection graphs of \emph{fat objects}. A geometric object $g \subset \real^d$ is said to be $\alpha$-fat for some $\alpha \ge 1$, if there exist balls $B_{\text{in}}, B_{\text{out}}$ such that $B_\text{in} \subseteq g \subseteq B_\text{out}$, such that the ratio of the radius of $B_\text{out}$ to that of $B_\text{in}$ is at most $\alpha$. We say that a set $F$ of objects is \emph{fat} if there exists a constant $\alpha \ge 1$ such that every geometric object in $F$ is $\alpha$-fat. Furthermore, we say that $F$ is a set of \emph{similarly-sized} fat objects, if the ratio of the largest diameter of an object in $F$, to the smallest diameter of an object in $F$ is at most a fixed constant. Finally, observe that if $F$ is a set of \emph{similarly-sized fat objects}, then the ratio of the largest out-radius to the smallest in-radius of an object is also upper bounded by a constant.  de Berg et al. \cite{BergBKMZ20} prove the following two results regarding the intersection graphs of similarly sized fat objects.

\begin{lemma}[\cite{BergBKMZ20}] \label{fv:lem:deberg-kpart}
	Let $d \ge 2$ be a constant. Then, there exist constants $\k$ and $\Delta$, such that for any intersection graph $G = (V, E)$ of an (unknown) set of $n$ similarly-sized fat objects in $\real^d$, a $\k$-partition $\P$ for which $G_\P$ has maximum degree $\Delta$ can be computed in time polynomial in $n$.
\end{lemma}

In the following, we will use the tuple $(G, d, \P, G_\P)$ to indicate that $G = (V, E)$ is the intersection graph of $n$ similarly-sized fat objects in $\real^d$, $\P$ is a \kpart of $G$ such that $\Gp$ has maximum degree $\Delta$, where $\k, \Delta$ are constants, as guaranteed by Lemma \ref{fv:lem:deberg-kpart}. 

\begin{lemma}[\cite{BergBKMZ20}] \label{fv:thm:deberg-weighted-treewidth}
	For any $(G, d, \P, \Gp)$, the weighted treewidth of $G_\P$ is  $\Oh(n^{\dd})$.
\end{lemma}

Now we are ready to prove the following result about the intersection graphs of similarly sized fat objects. This result is at the heart of the subexponential algorithms designed in the following sections.

\begin{theorem} \label{fv:thm:treedepth} There is a polynomial space algorithm that for a 
	given $(G, d, \P, \Gp)$, computes in time  $\tndd$  	
	 a weighted treedepth decomposition $(F, \varphi)$ of $G_\P$ of weighted treedepth  $\Oh(n^{\dd})$. 
\end{theorem}
\begin{proof}
	We use the approximation algorithm from \cite{reed2003algorithmic} to compute a weighted tree decomposition $(T, \beta)$ of $G_\P$ (see the later part of the proof for a detailed explanation). Using the standard properties of the tree decomposition (e.g., see \cite{cygan2015parameterized}), there exists a node $t \in V(T)$, such that $V_B \coloneqq \bigcup_{V_i \in \beta(t)} V_i$ is an $\alpha$-balanced separator for $G$, for some $\alpha \le 2/3$. Let $B \coloneqq \beta(t)$. Note that $B \subseteq \P$.
	
	Now we construct a part of the forest $F$, and the associated bijection $\varphi$ in the weighted treedepth decomposition $(F, \varphi)$ of $\Gp$.  We create a path $\pi = (u_1, u_2, \ldots, u_{|B|})$, and arbitrarily assign $\varphi(u_i)$ to some $V_i \in B$ such that it is a bijection. We set $u_1$, the first vertex on $\pi$, to be the root of a tree in $F$. We also set the weight $\omega(u_i) = \log(1+|V_i|)$, where $\varphi(u_i) = V_i$. Note that the $\omega(\pi) = \omega(B) = \ndd$.
	
	Let $(Y_1, Y_2)$ be the separation of $G$, corresponding to the separator $\bigcup_{V_i \in \beta(t)} V_i$. Analogously, let $(\P'_1, \P'_2)$ denote the separation of $\Gp$, corresponding to the separator $B$.	Furthermore, let $X_i \coloneqq Y_i \setminus V_B$, and $\P_i \coloneqq \P'_i \setminus B$ for $i = 1, 2$. Note that $X_1 \setminus V_B, X_2 \subseteq V(G)$ are disjoint, $\max\{|X_1|, |X_2|\} \le \alpha n$, and there is no edge from a vertex in $X_1$, to a vertex in $X_2$. Furthermore, $\P_1$ is a \kpart of $G[X_1]$, and $\P_2$ is a \kpart of $G[X_2]$.
	
	Now, we recursively construct weighted treedepth decomposition $(F_1, \varphi_1)$ of $G_\P[\P_1]$. Note that $\varphi_1$ is a bijection between $V(F_1)$ and $\P_1 \subseteq \P$. Let $R_1$ denote the set of roots of the trees in forest $F_1$. We add an edge from the last vertex $u_{|B|}$ on the path $\pi$, to each root in $R_1$. In other words, we attach every tree in $F_1$ as a subtree below $u_{|B|}$. The bijection $\varphi$ is extended to $\P_1$ using $\varphi_1$. Now we consider a weighted treedepth decomposition $(F_2, \varphi_2)$ of $\Gp[\P_2]$, and use it to extend $(F, \varphi)$ in a similar manner. This completes the construction of $(F, \varphi)$.
	
	Let us first analyze the weighted treedepth of $(F, \varphi)$. Let us use $q \coloneqq \dd$ for simplicity. For a path $\pi$ in $F$, let $\omega(\pi)$ denote the sum of weights of vertices along the path $\pi$. Recall that the  weight of any root-leaf path $\pi$ in $F$ is at most $\Oh(n^q)$. More generally, let $c' \ge 0$ be a universal constant (independent of the path $\pi$, or its level in $F$) such that the weight of a path corresponding to a separator computed at level $j$, is at most $c' \cdot (\alpha^{j-1}n)^q$. Since $\max\{|X_1|, |X_2|\} \le \alpha n$, we inductively assume that the weighted treedepth of $(F_1, \varphi_1)$, and that of $(F_2, \varphi_2)$ is at most $\Oh(\alpha^q \cdot n^q)$. More specifically, we assume that there exists a universal constant $c \ge c'$, such that the sum of the weights along any root-leaf path in $F_1$ is upper bounded by $c \cdot \frac{(\alpha n)^q}{1-\alpha^q}$. The same inductive assumption holds for any root-leaf path in $F_2$. Therefore, the weight of any root-leaf path in $F$ is upper bounded by 	
	$$\omega(\pi) + c \cdot \frac{\alpha^q n^q}{1-\alpha^q} \le cn^q \lr{1+ \frac{\alpha^q}{1-\alpha^q}} =  \frac{cn^q}{1-\alpha^q}.$$ 
	Therefore, we have the desired bound on the weighted treedepth by induction. 
	
	Now we look the treewidth construction part of the algorithm in order to sketch the claims about bounds on time and space. Given the graph $G_\P$, we construct a graph $H$ by replacing every vertex $V_i$ with a (new) clique $C_i$ of size $\log(1+|V_i|)$. If $V_iV_j \in E(G_\P)$, we also add edges from every vertex in $C_i$ to every vertex in $C_j$. As shown in \cite{BergBKMZ20}, the weighted width of $G_\P$ is equal to the treewidth of $H$, plus $1$. Note that $|V(H)| = \sum_{V_i \in \P}\log(1+|V_i|) \le n$, since $\P$ is a partition of $V(G)$. 
	
	The algorithm from \cite{reed2003algorithmic} (see also Section 7.6.2 in the Parameterized Algorithms book \cite{cygan2015parameterized}) for approximating treewidth of a graph $H$ works as follows. Suppose the treewidth of a graph is $k$, which is known. At the heart of this algorithm is a procedure $\mathtt{decompose}(W, S)$, where $S \subsetneq W \subseteq V(H)$, and $|S| \le 3k+4$. This procedure tries to decompose the subgraph $H[W]$ in such a way that $S$ is completely contained in one bag of the tree decomposition. The first step is to compute a partition $(S_A, S_B)$ of $S$, such that the size of the separator separating $S_A$ and $S_B$ in $H[W]$ is at most $k+1$. This is done by exhaustively guessing all partitions, which takes $2^{\Oh(k)}$ time. For each such guess of $(S_A, S_B)$, we run a polynomial time algorithm to check whether the bound on the separator size holds. Once such a partition is found, a set $\hat{S} \supsetneq S$ is found by augmenting $S$ in a particular way. Finally, we recursively run the procedure $\mathtt{decompose}(N_H[D], N_H(D))$, for each connected component $D$ in $H[W \setminus \hat{S}]$. Finally, the tree decomposition of $H[W]$ is computed by augmenting the tree decompositions computed by the recursive procedure for its children, with the root bag containing $\hat{S}$. It is shown that this algorithm computes a tree decomposition of width $\Oh(\tw)$ in time $2^{\Oh(\tw)} \cdot n^{\Oh(1)}$. Furthermore, it can also be observed that it only uses polynomial space.
	
	Therefore, computing a tree decomposition of $\Gp$ of weighted treewidth $\ndd$ takes $\tndd$ time and polynomial space, corresponding to the original graph $G$ with $n$ vertices. The treewidth computation algorithm is called at most $n$ times, and there is additional polynomial processing at every step. This implies the time and space bounds as claimed.
\end{proof}

We note that de Berg et al.\ \cite{BergBKMZ20} show the existence of an $\alpha$-balanced separator of weight $\ndd$, which they then use to show the same bound on weighted treewidth (Theorem \ref{fv:thm:deberg-weighted-treewidth}). Moreover, this separator can be computed in $\Oh(n^{d+2})$ time if we are also given the geometric representation of the underlying objects in $\real^d$. However, without geometric representation it is not clear whether this separator can be directly computed. Therefore, we first compute an approximate weighted treedepth decomposition, and then retrieve the separator bag in the proof of Theorem \ref{fv:thm:treedepth}. Now, we prove the following abstract theorem which models at a high level our subexponential algorithms using polynomial space.

\begin{theorem} \label{fv:thm:subexp-poly}
	Let $\mathcal{A}$ be an algorithm for solving a problem on graph $G$, that takes input $(G, d, \P, \Gp)$, and a weighted treedepth decomposition $(F, \varphi)$ of $\Gp$ of weighted depth $\ndd$ (and optionally additional inputs of polynomial size). Suppose $\mathcal{A}$ is a recursive algorithm which uses $(F, \varphi)$ in the following way. At every node $u \in V(F)$, it spends time proportional to $2^{\Oh(\omega(u))} \cdot n^{\Oh(1)}$, uses polynomial space, and makes at most $2^{\Oh(\omega(u))}$ recursive calls on the children of $u$. Then, the algorithm $\mathcal{A}$ runs in time $\tndd$, and uses polynomial space.
\end{theorem}
\begin{proof}
	In the following, to avoid cumbersome notation, we assume that the constants in the exponent of $2^{\Oh(\omega(u))}$ are $1$, i.e., that the algorithm spends $2^{\omega(u)} \cdot n^{\Oh(1)}$ time at $u$, and makes $2^{\omega(u_i)}$ recursive calls. It is easy to see that any other constants in the exponent is absorbed in the exponent of $\tndd$.
	
	Consider node $u \in V(F)$. We will inductively prove that the overall running time of $\mathcal{A}$ corresponding to one recursive call on $u$ is upper bounded by $\displaystyle \ell(u)  \cdot n^{\Oh(1)} \cdot \max_{\pi_u} 2^{2\omega(\pi_u)} $, where $\ell(u)$ is the number of leaves in the subtree of $F$ rooted at $u$, and the maximum is taken over all root-leaf paths $\pi_u$ in the subtree rooted at $u$. Assuming this claim is true, then the bound on the running time holds as follows. The running time of $\mathcal{A}$ corresponding to the root $r$ of $F$ is upper bounded by $n \cdot n^{\Oh(1)} \cdot \max_{\pi} 2^{2\omega(\pi)}$, where the maximum is taken over all root-leaf paths in $F$. However, $2^{2\omega(\pi)}$ is at most $2^{\Oh(\wtd(\Gp))}$, which is at most $\tndd$. Furthermore, if the algorithm uses polynomial space at every node in $V(F)$, and since the (unweighted) depth of $F$ is at most $n$, it only uses polynomial space overall.
	
	Now we prove the inductive claim. Fix a node $u \in V(F)$. If $u$ is leaf in the tree, then the running time is upper bounded by $2^{\omega(u)} \cdot n^{\Oh(1)}$ by the property of $\mathcal{A}$. Note that $\ell(u) = 1$ in this case. Now, suppose $u$ is an arbitrary internal node in $V(F)$. The algorithm spends time proportional to $2^{\omega(u)} \cdot n^{\Oh(1)}$, and makes at most $2^{\omega(u)}$ recursive calls on the children of $u$. Let $C(u)$ denote the set of children of $u$. By the inductive hypothesis, time taken by $\mathcal{A}$ in one recursive call at any child $v_i \in C(u)$ is bounded by $\displaystyle \ell(v_i)  \cdot n^{\Oh(1)} \cdot \max_{\pi_{v_i}} 2^{2\omega(\pi_{v_i})}$, where the maximum is taken over all $v_i$-leaf paths $\pi_{v_i}$ in the subtree rooted at $v_i$. Therefore, the overall running time at $u$ is bounded by
	\begin{align*}
		T(u) &\le 2^{\omega(u)} \cdot n^{\Oh(1)} + 2^{\omega(u)} \cdot \sum_{v_i \in C(u)} \ell(v_i) \cdot n^{\Oh(1)} \cdot \max_{\pi_{v_i}} 2^{2\omega(\pi_{v_i})}
		\\&\le n^{\Oh(1)} \cdot \lr{2^{\omega(u)} + \sum_{v_i \in C(u)} \ell(v_i) \cdot \max_{\pi_{v_i}} 2^{\omega(u) + 2\omega(\pi_{v_i})}}
		\\&\le \ell(u) \cdot n^{\Oh(1)} \cdot \lr{2^{\omega(u)} + \max_{\pi_u} 2^{\omega(u) + 2\omega(\pi_{v_i})}} \tag{Since there are at most $\ell(u)$ paths $\pi_u$ going from $u$ to a leaf}
		\\&\le \ell(u) \cdot n^{\Oh(1)} \cdot \max_{\pi_u} 2^{2\omega(\pi_u)}.
	\end{align*}
\end{proof}

\iflong
\section{Simple Recursive Algorithms} \label{fv:sec:recursive}

\subsection*{Independent Set}
Recall that the task of the \textsc{Independent Set} problem is, given a graph $G$, to find an independent set, i.e., a set of pairwise nonadjacent vertices, of maximum size.
Given $G = (V, E)$, we first obtain $(G, d, \P, \Gp)$, and associated $(F, \varphi)$ of weighted treedepth $\ndd$ in $\tndd$ time and polynomial space using Theorem \ref{fv:thm:treedepth} We have the following simple observation from \cite{BergBKMZ20}.

\begin{observation}[\cite{BergBKMZ20}] \label{fv:obs:indepset-kpart}
	Let $\P$ be a \kpart of $G = (V, E)$, and let $I \subseteq V(G)$ be an independent set. Then, $|I \cap V_i| \le \k$ for any $V_i \in \P$.
\end{observation}

We now describe our recursive algorithm. In addition to $(G, d, \P, \Gp)$, and $(F, \varphi)$, the algorithm takes input $u \in V(F)$, and $D \subseteq V(G)$. Here, $u$ is the current node in $V(F)$ in the recursive algorithm, and $D \subseteq V(G)$ denotes the set of \emph{disallowed} vertices that cannot be added to the independent set, due to previous choices made by the algorithm. 

Let $\varphi(u) = V_i \in \P$ be the associated part in the \kpart $\P$, and define $\I_u$ be the collection of subsets $I \subseteq V_i$ such that (1) $I$ is independent in $G$, (2) $I \cap F = \emptyset$, and (3) $|I| \le \k$. For each guess $I \in \I_u$, we call the algorithm recursively on the children of $u$, where the set of disallowed nodes is updated to $D' \gets D \cup N_G(I)$. For a child $v_j$ of $u$, let $\textsc{MaxIS}(v_j, I)$ denote the independent set returned by the algorithm corresponding to guess $I$. We return the independent set maximizing $|I \cup \bigcup_{v_j \in C(u)} \textsc{MaxIS}(v_j, I)|$ over all $I \in \I_u$. It is easy to see the correctness of the algorithm. It is also easy to extend the algorithm to \textsc{Weighted Independent Set}, where the input is a weighted graph and the task is to find an independent set of maximum weight.

Note that $|\I_u| \le (1+|V_i|)^\k = \exp(O(\k \log(1+|V_i|))) = 2^{O(\omega(u))}$, since $\k = O(1)$. Therefore, the algorithm spends at most $2^{O(\omega(u))} \cdot n^{O(1)}$ time at $u$, uses polynomial space, and corresponding to each of the $2^{O(\omega(u))}$ guesses for $I$, it calls the algorithm recursively on its children. Using Theorem \ref{fv:thm:subexp-poly}, the bounds on space and time follow. 

\begin{theorem} \label{fv:thm:max-is}
	There exists a $\tndd$ time, polynomial space algorithm to compute a maximum (weight) independent set in the intersection graphs of similarly sized fat objects in $\real^d$.
\end{theorem}

\begin{figure}[h]
	\centering
	\begin{tabular}{|c|c|c|c|c|c|}
		\hline
		Algorithm & Space & Similarly-sized & Convex & Robust \\
		\hline
		Corollary 2.4 in \cite{BergBKMZ20} & {\color{ForestGreen}Poly} & {\color{ForestGreen}Yes} & {\color{ForestGreen}Yes} & {\color{Red}No} \\
		\hline
		Theorem 2.13 in \cite{BergBKMZ20} & {\color{red}$\tndd$} & {\color{ForestGreen}Yes} & {\color{Red}No} & {\color{ForestGreen}Yes} \\
		\hline
		Theorem \ref{fv:thm:max-is} & {\color{ForestGreen}Poly} & {\color{ForestGreen}Yes} & {\color{Red}No} & {\color{ForestGreen}Yes} \\
		\hline
	\end{tabular}
	\caption{Comparison of three $\tndd$-time algorithms for \textsc{Independent Set} on the intersection graphs of fat objects. ``Robust'' in the last column means that the algorithm does not require the geometric representation of the objects. Our algorithm (third row) can be seen as a re-interpretation of the one in the first row, using the framework of weighted treedepth. }
\end{figure}

\subsection*{$r$-Dominating Set}
For a positive integer $r$, a set of vertices $D\subseteq V(G)$ is an \emph{$r$-dominating} set if every vertex of $G$ is at distance at most $r$ from some vertex of $D$. Throughout this subsection we assume that $r$ is a fixed constant. The \textsc{$r$-Dominating Set} asks for an $r$-dominating set in a given graph.

\begin{claim}[\cite{BergBKMZ20}]
	Suppose that $G = (V, E)$ has a \kpart $\P$ such that $G_\P$ has maximum degree $\Delta$. Then, there exists an $r$-dominating set $D \subseteq V$ such that $|D \cap V_i| \le \k^2 (\Delta+1)$ for any $V_i \in \P$.
\end{claim}

At a high level, suppose there exists an $r$-dominating set $D$ such that for some $V_i \in \P$, $|D \cap V_i| > \k^2 (1+\Delta)$. Recall that each $V_i$ is union of at most $\k$ cliques. Therefore, there exists a clique $C \subseteq V_i$, such that $|C \cap D| > \k (1+\Delta)+1$. Then, we remove all but one vertex of $C \cap D$, and replace them with one vertex each in $\k$ cliques $C' \subseteq V_j$, where $V_iV_j \in E(G_\P)$. Note that $|D'| \le |D|$, using the properties of $(G, d, \P, \Gp)$, it can be shown that $D'$ is also $r$-dominating set of $G$. We repeat this operation until the condition is satisfied for every $V_i \in \P$. A formal proof can be found in \cite{BergBKMZ20}.

For \textsc{$r$-Dominating Set}, we cannot directly use Theorem \ref{fv:thm:subexp-poly} to show the bound on running time. In fact, we need the following strengthened version of Lemma \ref{fv:thm:deberg-weighted-treewidth} from \cite{BergBKMZ20}, which we prove for completeness. We need the following notation. For any $V_i \in \P$, let $N_r[V_i] \coloneqq \{V_i\} \cup \{V_j \in \P: d_{G_\P}(V_i, V_j) \le r\}$.

\begin{theorem}[\cite{BergBKMZ20}] \label{fv:thm:deberg-rdom-kpart}
	Let $G$ be the intersection graph of $n$ similarly-sized fat objects. Then, there exists a \kpart $\P$ and a corresponding $G_\P$ of maximum degree at most $\Delta$, where $\k, \Delta$ are constants, and $G_\P$ has a weighted tree decomposition $(T, \beta)$ with the additional property that for any node $t$, the total weight of the partition classes in $\bigcup_{V_i \in \beta(t)} \LR{N_r[V_i]}$ is $\ndd$.
\end{theorem}
\begin{proof}
	Let $\P$ be a $\kappa$-partition of $G = (V, E)$ as guaranteed by Lemma \ref{fv:lem:deberg-kpart}, such that $G_\P$ has maximum degree $\Delta$. For every $V_i$, we define a set $W_i$, where we add a copy of a geometric object belonging to any $V_j \in N_r[V_i]$ to $W_i$. We say that $V_i$ is the core of $W_i$. Note that there are at most $c = O(\Delta^r) = O(1)$ copies of every original geometric object. 
	
	Now, let $G^r$ be the intersection graph defined by the set of geometric objects $\bigcup_{V_i \in \P} W_i$. Note that $G^r$ is the intersection graph of similarly sized fat objects in $d$ dimensions, and the number of vertices in $G^r$ is at most $nc = O(n)$. Furthermore, $\P^r = \{W_i: V_i \in \P\}$ is a $c\k$-partition of $G^r$ with maximum degree $\Delta$. Therefore, using Lemma \ref{fv:thm:deberg-weighted-treewidth}, we conclude that the weighted treewidth of $G^r$ is $\ndd$. Let $(T^r, \beta^r)$ be a weighted tree decomposition of $G^r$ with width $\ndd$. Then, we obtain a weighted tree decomposition $(T, \beta)$ of $G_\P$ by replacing each $W_i$ with its core $V_i$, and observe that $\P$ and $G_\P$ have the claimed properties.
\end{proof}

Now we use the additional guarantee of Theorem \ref{fv:thm:deberg-rdom-kpart} to obtain a $\tndd$ time, polynomial space algorithm for \textsc{$r$-Dominating Set}. One important difference in this algorithm as compared to the other algorithms, is that instead of using weighted treewidth decomposition via Theorem \ref{fv:thm:treedepth}, and then appealing to Theorem \ref{fv:thm:subexp-poly}, it is convenient to state and analyze the algorithm via the separator tree $\Sigma$, defined as follows. Each node of $\Sigma$ is an $\alpha$-balanced separator of an induced subgraph of $G$. If $\ts \subseteq V(G)$ is a node at level $j$ in $V(\Sigma)$, then we will have the following properties (proved subsequently). 
\begin{enumerate}
	\item $\ts$ is equal to $\bigcup_{V_i \in \P'} V_i$ for some subset $\P' \subseteq \P$. Denote this $\P'$ by $\P(\ts)$.
	\item Let $\Pi(\ts) \subseteq V(G)$ denote the union of all ancestors $\ts'$ of $\ts$ \footnote{We assume that a node is an ancestor and a descendant of itself.}, and let $J(\ts) \coloneqq G[\Pi(\ts)]$.
	\item Let $\C(\ts) \subseteq V(G)$ denote the union of all descendants $\ts'$ of $\ts$ in $\Sigma$. Then, $\ts$ is an $\alpha$-balanced separator of $H(\ts) \coloneqq G[\C(\ts)]$. 
	\item $\ts$ induces a separation $(\widetilde{Y}_1, \widetilde{Y}_2)$ in $H(\ts)$, such that $\max\{|\widetilde{X_1}|, |\widetilde{X_2}|\} \le \alpha^{j} n$, where $\widetilde{X}_i = \widetilde{Y}_i \setminus S$ for $i = 1, 2$. Furthermore, $\ts$ has two children $\ts_1, \ts_2$ in $\Sigma$, which are $\alpha$-balanced separators of $G[\widetilde{X}_1], G[\widetilde{X}_2]$, respectively.
	\item Let $\mathcal{N}(\ts) \subseteq \P$ be the set of all $V_j \in \P$, such that (1) $V_j \in N_r[V_i]$, and (2) $V_j \subseteq \C(\ts)$. Then, the total weight of $\mathcal{N}(\ts)$ is upper bounded by $O((\alpha^{j-1} \cdot n)^{\dd})$ (via the guarantee in Theorem \ref{fv:thm:deberg-rdom-kpart}).
\end{enumerate}

We will construct the tree $\Sigma$ recursively. As for the base case, we compute the weighted tree decomposition $(T, \beta)$ of $G_\P$ with the additional property from Theorem \ref{fv:thm:deberg-rdom-kpart}, using the algorithm of \cite{reed2003algorithmic,cygan2015parameterized}. As argued in Theorem \ref{fv:thm:treedepth}, this takes $\tndd$ time and polynomial space. Let $t \in V(T)$ be the separator node, such that $S \coloneqq \bigcup_{V_i \in \beta(t)} V_i$ is an $\alpha$-balanced separator for $G$, for some $\alpha \le 2/3$. It is easy to see that Property 1 holds -- $\P(S) = \beta(t)$. Property 2 is only a definition. Furthermore, Property 5 holds due to the additional guarantee from Theorem \ref{fv:thm:deberg-rdom-kpart}. Finally, let $(Y_1, Y_2)$ be the separation of $G$ due to the separator $S$, and let $X_i = Y_i \setminus S$ for $i = 1, 2$. We will compute $\alpha$-balanced separators $S_1, S_2$ of $G[X_1], G[X_2]$ respectively, and let $S_1, S_2$ to be the two children of $S$ in the tree $\Sigma$. Since $\max\{|X_1|, |X_2|\} \le \alpha n$, Property 4 holds.

Consider a general node $\ts \subseteq V(G)$ at $j$-th level in $\Sigma$, we may inductively assume that $\ts$ is an $\alpha$-balanced separator of an induced subgraph $H$ of $G$, where $H$ contains at most $\alpha^{j-1} n$ vertices. Then, by using arguments similar to the previous paragraph, it follows that the Properties 1-5 hold for $\ts$. After the completing the construction of $\Sigma$, it is apparent that every $v \in \C(\ts)$ is contained in exactly one descendant $\ts'$ of $\ts$. Therefore, $H = H(\ts)$. This completes the construction of the separator tree $\Sigma$. Now let us explain how to compute an $r$-dominating set using this.

Our algorithm is recursive. In addition to $G, \P, \Gp$, and $\Sigma$, it has additional two inputs: $\ts \in V(\Sigma)$, the current node in the separator tree that the algorithm is operating at; and $L \subseteq \Pi(\ts)$, which corresponds to the subset of vertices that are already added to the dominating set in the current recursive call. Furthermore, a $V_j \in \mathcal{N}(\ts)$ may be ``marked'', which indicates that the algorithm has already fixed an $O(\k^2 \Delta)$-size subset of $V_j$ in $L$. 

Let $\ts' = \bigcup_{V_j \in \mathcal{N}(\ts)} V_j$. Now, let $\mathcal{D}$ be the collection of subsets $D \subseteq \ts$ with the following properties: (1) $L \cup D$ is a dominating set of $J(\ts)$, (2) $|D \cap V_j| \le \k^2 (\Delta+1)$ for any $V_j \in \mathcal{N}(\ts)$, and (3) If $V_j \in \mathcal{N}(\ts)$ was already marked, then $D \cap V_j = L \cap V_j$.
Note that the number of such subsets is upper bounded by 
\begin{align*}
	\prod_{V_j \in \mathcal{N}(\ts)} (1+|V_j|)^{\k^2 (1+\Delta)} &= \exp\lr{\k^2 (1+\Delta) \sum_{V_j \in \mathcal{N}(\ts)} \log(1+|V_j|)} 
	\\&= \exp\lr{O(\alpha^{j-1} \cdot n)^{\dd}} \tag{Using Property 5}
\end{align*}

Furthermore, we ``mark'' $V_j \in \mathcal{N}(\ts)$, indicating that we have already fixed the choice of dominating set from $V_j$. Observe that this restriction can only reduce the number of recursive calls made at a subsequent level. Therefore, the number of recursive calls made to the children of $\ts$ can be upper bounded by $\exp\lr{O(\alpha^{j-1} \cdot n)^{\dd}}$, which implies that recurrence for the running time is given by
\begin{align*}
	T(\alpha^{j-1} n) \le \exp\lr{O(\alpha^{j-1} \cdot n)^{\dd}} \cdot n^{O(1)} + \exp\lr{O(\alpha^{j-1} \cdot n)^{\dd}} \cdot 2 \cdot T(\alpha^j n)
\end{align*}
This recurrence solves to $T(n) = \tndd$. 

It is easy to see that the construction of the separator tree $\Sigma$ only requires polynomial space, which follows from a similar discussion in Theorem \ref{fv:thm:treedepth}. Furthermore, the recursive algorithm uses polynomial space at every node of $\Sigma$, and there are $O(n)$ nodes in $\Sigma$. Thus, the overall space complexity of the algorithm is polynomial in $n$.

\begin{theorem} \label{fv:thm:min-ds}
	For any fixed $r \ge 1$, there exists a $\tndd$ time, polynomial space algorithm to compute a minimum $r$-dominating set in the intersection graphs of similarly sized fat objects in $\real^d$.
\end{theorem}
\fi
\section{\cnc Algorithms} \label{fv:sec:cnc}
Hegerfeld and Kratsch \cite{HegerfeldK20} adapt the \cnc technique to give FPT algorithms for various connectivity based subset problems, parameterized by (unweighted) treedepth. In particular, these algorithms are randomized, have running times of the form $2^{\Oh(d)} \cdot n^{\Oh(1)}$, and use polynomial space. In their work, they consider \textsc{Connected Vertex Cover, Feedback Vertex Set, Connected Dominating Set, Steiner Tree,} and \textsc{Connected Odd Cycle Transversal} problems. We are able to adapt their technique for all of these problems, except for \textsc{Connected Dominating Set}. For the rest of the problems, we will extend their ideas to the more general case of \emph{weighted treedepth}, and use it to give $\tndd$ time, polynomial space, randomized algorithms. In the following, we select \textsc{Steiner Tree} as a representative problem, and we will explain  in full detail. For the remaining problems, we will only give a brief sketch of the differences, since at a high level the approach remains the same.

\subsection{Setup}

We adopt the following notation from \cite{HegerfeldK20}.

Let $\cc(G)$ denote the number of connected components in $G$. A cut of $X \subseteq V(G)$ is a pair $\lrcut$, where $X_L \cap X_R = \emptyset$, $X_L \cup X_R = X$. We refer to $X_L, X_R$ as the left and the right side of the cut $\lrcut$ respectively. A cut $\lrcut$ of $G[X]$ is consistent, if for any $u \in X_L$ and $v \in X_R$, $uv \not\in E(G[X])$. A \emph{consistently cut subgraph} of $G$ is a pair $(X, \lrcut)$, such that $X \subseteq V(G)$, and $\lrcut$ is a consistent cut of $G[X]$. Finally, for $X \subseteq V(G)$, we denote the set of consistently cut subgraphs of $G[X]$ by $\C(X)$.

For $n \in \natural$, let $[n]$ denote the set of integers from $1$ to $n$. For integers $a, b$, we write $a \equiv b$ to indicate equality modulo $2$. We use Iverson's bracket notation: for a boolean predicate $p$, $[p]$ is equal to $1$ if $p$ is true, otherwise $[p]$ is equal to $0$. 

Consider a function $f: A \to S$. For every $s \in S$ and a set $X$, we define the set $X(f, s) \coloneqq X \cap f^{-1}(\{s\})$ -- note that $X(f, s)$ may be empty for some or all $s \in S$. Furthermore, observe that the sets $\{A(f, s)\}_{s \in S}$ define a partition of $A$. For two functions $g: A \to S$, $f: B \to S$, we define the new function $g \oplus f: (A \cup B) \to S$ as follows. $(g \oplus f)(e) = f(e)$ for $e \in B$, and $(g \oplus f)(e) = g(e)$ for $e \in (A \setminus B)$. That is, $(g \oplus f)$ behaves like $g$ and $f$ on the exclusive domains, but in case of a conflict, the function $f$ takes the priority.

Recall that we work with $(G, d, \P, \Gp)$, and the corresponding weighted treedepth decomposition $(F, \varphi)$ of $G$. Here, $\varphi$ is a bijection between $V(F)$ and $\P$. For a node $u_i$, we will use $V_i \coloneqq \varphi(u_i)$, i.e., we use the same indices in the subscript to identify a node of $F$ and the corresponding part in $\P$. We denote the set of children of $u_i$ by $\child(u_i)$. We also define the following sets.
\begin{align*}
	\tail[u_i] &= \bigcup_{u_j \text{ is an ancestsor of } u_i} V_j\ ; && \tail(u_i) = \tail[u_i] \setminus V_i
	\\\tree[u_i] &= \bigcup_{u_j \text{ is a descendant of } u_i} V_j\ ; && \tree(u_i) = \tree[u_i] \setminus V_i
	\\\broom[u_i] &= \tail[u_i] \cup \tree(u_i)
\end{align*} 

\paragraph{Isolation Lemma.} 
\begin{definition}
	Let $U$ be a finite set, and $\mathcal{F} \subseteq 2^U$ be a family of subsets of $U$. We say that a weight function $\w: U \to \integer$ \emph{isolates} the family $\mathcal{F}$ if there exists a unique set $S' \in \mathcal{F}$ such that $\w(S') = \min_{S \in \mathcal{F}} \w(S)$, where $\w(X) \coloneqq \sum_{x \in X} \w(x)$ for any subset $X \subseteq U$.
\end{definition}
The following isolation lemma due to Mulmuley et al.\ \cite{mulmuley1987matching} is at the heart of all \cnc algorithms.
\begin{lemma}[\cite{mulmuley1987matching}] \label{fv:lem:isolation}
	Let $\mathcal{F} \subseteq 2^U$ be a non-empty family of subsets of a finite ground set $U$. Let $N \in \natural$, and suppose $\w(u)$ is chosen uniformly and independently at random from $[N]$ for every $u \in U$. Then, $\Pr(\w \text{ isolates } \mathcal{F}) \ge 1 - |U|/N$.
\end{lemma}

\paragraph{General Idea.}

Fix a problem involving connectivity constraints that we want to solve on $G$. Note that the connectivity constraints may be explicit, e.g., \textsc{Connected Vertex Cover, Steiner Tree}, or implicit, e.g., \textsc{Odd Cycle Transversal}. Let $U$ be the ground set that is related to the graph $G$, such that $\S \subseteq 2^U$, where $\S$ denotes the set of solutions to the problem. At a high level, a \cnc based algorithm contains the following two parts.
\begin{itemize}
	\item \textbf{The Cut part:} We obtain a set $\R$ by relaxing the connectivity requirements on the solutions, such that $\S \subseteq \R \subseteq 2^U$. The set $\Q$ will contain pairs $(X, C)$, where $X \in \R$ is a candidate solution, and $C$ is a consistent cut of $X$. Note that since $X \in \R$, $X$ may be possibly disconnected. 
	\item \textbf{The Count part:} We compute $|\Q| \mod 2$ using an algorithm. The consistent cuts are defined carefully, in order that the non-connected solutions from $\R\setminus \S$ cancel while counting modulo $2$, since they are consistent with an even number of cuts. 
\end{itemize}
Note that if $|\S|$ is even, then the procedure counting $|\Q| \mod 2$ will return $0$, which will be inconclusive. Therefore, we initially sample a random weight function $\w: U \to [N]$ for some large integer $N \ge 2|U|$, and count $|\Q_w| \mod 2$ (where $\Q_w$ is the subset of $\Q$ such that the corresponding $X$ has weight \emph{exactly} $w$), for all values of $w \in [2|U|^2]$. Using Lemma \ref{fv:lem:isolation}, it can be argued that with at least probability $1/2$, if $\S \neq \emptyset$, then for some weight $w \in [2|U|^2]$, the procedure counting $|\Q_w| \mod 2$ outputs $1$. Finally, we guess an arbitrary vertex $v_1 \in V(G)$ in the solution, and force it to be on the left side of the consistent cuts. That is, we count the number of consistent cuts in which $v_1$ is forced to belong to the left side. This breaks the left-right symmetry. We first have the following two results from \cite{HegerfeldK20,cygan2011solving}. 

\begin{lemma}[\cite{HegerfeldK20,cygan2011solving}] \label{fv:lem:num-ccs}
	Let $X \subseteq V(G)$ such that $v_1 \in X$. The number of consistently cut subgraphs $(X, \lrcut)$ such that $v_1 \in X_L$ is equal to $2^{\cc(G[X])-1}$. 
\end{lemma}

\begin{corollary}[\cite{HegerfeldK20,cygan2011solving}] \label{fv:cor:cnc}
	Let $\S \subseteq 2^U$, and $\Q \subseteq 2^{U \times (V \times V)}$, such that for every $\w: U \to [2|U|]$, and a target weight $w \in [2|U|^2]$, the following two properties hold.
	\begin{enumerate}
		\item $|\LR{(X, C) \in \Q : \w(X) = w}| = |\LR{X \in \S : \w(X) = w}|$, and
		\item There is an algorithm $\texttt{CountC}(\w, w, (G, d, \P, \Gp), (F, \varphi))$, where $(F, \varphi)$ is a weighted treedepth decomposition of $(G, d, \P, \Gp)$, such that:
		$\texttt{CountC}(\w, w, (G, d, \P, \Gp), (F, \varphi)) \equiv |\{(X, C \in \Q | \w(X) = w)\}|$	
	\end{enumerate}
	\iflong
	Then, Algorithm \ref{fv:alg:cnc} returns \textbf{false} if $\S = \emptyset$, and returns \textbf{true} with probability at least $\frac{1}{2}$ otherwise.
	\fi
\end{corollary}
\iflong
\begin{proof}
	Plugging in $\mathcal{F} = \mathcal{S}$ and $N = 2|U|$ in Lemma \ref{fv:lem:isolation}, we know that if $\S \neq \emptyset$, then with probability at least $1/2$, there exists a weight $w \in [2|U|^2]$ such that $|\LR{X \in \S : \w(X) = w}| = 1$. Then, Algorithm \ref{fv:alg:cnc} returns \textbf{true} with probability at least $1/2$.
	
	On the other hand, if $\S = \emptyset$, then by the first property, and the definition of \texttt{CountC}, for any choice of $\w$ and $w$, the procedure $\texttt{CountC}$ returns \textbf{false}. Therefore, Algorithm \ref{fv:alg:cnc} returns \textbf{false}.
\end{proof}
\begin{algorithm}
	\caption{Cut\&Count$(U, (G, d, \P, \Gp), (F, \varphi), \texttt{CountC})$} \label{fv:alg:cnc}
	\begin{algorithmic}[1]
		\Statex \textbf{Input}: A set $U$, $(G, d, \P, \Gp)$, associated weighted treedepth decomposition $(F, \varphi)$, a procedure \texttt{CountC} that takes $\w: U \to [N], w \in \natural$
		\State Choose $\w(u)$ independently and uniformly at random from $[2|U|]$ for each $u \in U$
		\For{$w = 1, 2, \ldots, 2|U|^2$} 
		\State \textbf{if} $\texttt{CountC}((G, d, \P, \Gp), (F, \varphi), \w, w) \equiv 1$ \Return \textbf{true}
		\EndFor
		\State \Return \textbf{false}
	\end{algorithmic}
\end{algorithm}
\fi

\subsection{Steiner Tree}

\begin{definition}[\textsc{Steiner Tree}]
	\ \\Input: An undirected graph $G = (V, E)$, a set of terminals $K \subseteq V(G)$, and an integer $k$. 
	\\Question: Is there a subset $X \subseteq V(G)$, with $|X| \le k$, such that $G[X]$ is connected, and $K \subseteq X$?
\end{definition}

Fix a $(G, d, \P, \Gp)$ via Lemma \ref{fv:lem:deberg-kpart}. Recall that $\P$ is a \kpart of $G$, such that the corresponding graph $\Gp$ has maximum degree $\Delta = O(1)$. We first have the following lemma from \cite{BergBKMZ20}.

\begin{lemma}[\cite{BergBKMZ20}] \label{fv:lem:steiner-bound}
	Suppose $X$ is a minimal solution for \textsc{Steiner Tree} (i.e., no proper subset of $X$ is also a solution) for a given $(G, d, \P, \Gp)$, and a set of terminals $K$. Then $|X \cap (V_i \setminus K)| \le \k^2 (\Delta+1)$ for any $V_i \in \P$.
\end{lemma}

Let $k' = |K| + \k^2 (\Delta+1) \cdot |\P|$. Note that using Lemma \ref{fv:lem:steiner-bound}, we may assume that $k \le k'$ -- if $k \ge k'$, then $(G, k)$ is a ``yes-instance'' iff $(G, k')$ is a ``yes-instance''. For any $X \subseteq V(G)$, we say that $X$ is $\P$-restricted if for any $V_i \in \P$, $|X \cap (V_i \setminus K)| \le \k^2 (\Delta+1)$. Note that this definition of a $\P$-restricted set (and later, that of a $\P$-restricted function) is specific to the \textsc{Steiner Tree} problem. For different problems, we need to define this notion differently, albeit the main idea is to use a problem-specific version of Lemma \ref{fv:lem:steiner-bound}. 

We will run the following algorithm for all values of $k \le k'$. Let $t_1 \in K$ be an arbitrary terminal that we will fix to be on the left side of consistent cuts, as discussed previously. Now we give the formal definitions of the sets $\R, \S, \Q$ that were abstractly defined in the setup. We also define weight-restricted versions $\R_w, \S_w, \Q_w$ of these sets, where $w \in \natural$.
\begin{align*}
	\R &= \LR{X \subseteq V(G): X \text{ is $\P$-restricted}, K \subseteq X, \text{ and } |X| = k}; &&\R_w = \LR{X \in \R: \w(X) = w}
	\\\S &= \LR{X \in \R : G[X] \text{ is connected}};  &&\S_w = \LR{X \in \S: \w(X) = w}
	\\\Q &= \LR{(X, \lrcut) \in \C(V) : X \in \R \text{ and } t_1 \in X_L}; &&\Q_w = \{(X, \lrcut) \in \Q: \w(X) = w\}
\end{align*}

\begin{lemma} \label{fv:lem:sol-cuts}
	Let $\w: V(G) \to [N]$ be a weight function. Then, for every $w \in \natural$, $|\S_w| \equiv |\Q_w|$.
\end{lemma}
\begin{proof}
	From Lemma \ref{fv:lem:num-ccs}, $|\Q_w| = \sum_{X \in \R_w} 2^{\cc(G[X])-1}$. Therefore, $|\Q_w| \equiv |\LR{X \in \R_w : \cc(G[X]) = 1}| = |\S_w|$. Recall that $\equiv$ is equality modulo $2$.
\end{proof}

The goal of the rest of this subsection is to explain how the procedure \texttt{CountC} works.

First, we drop the cardinality constraints and define the following candidates and candidate cut-pairs for induced subgraphs $G[V']$, where $V' \subseteq V(G)$.
\begin{align*}
	\hat{\R}(V') &= \LR{X \subseteq V' : X \text{ is $\P$-restricted, and } K \cap V' \subseteq X }
	\\\hat{Q}(V') &= \LR{(X, \lrcut) \in \C(V') : X \in \R(V') \text{ and }  t_1 \in V' \implies\ t_1 \in X_L }
\end{align*}

Next, we define an important notion of $\P$-restricted functions, which will be crucial for pruning the number of recursive calls.
\begin{definition} \label{fv:defn:p-restricted}
	Let $f: X \to \states$ be a function, where $X \subseteq V(G)$. We say that $f$ is $\P$-restricted, if the following properties hold:
	\begin{itemize}
		\item $f^{-1}(\LR{\oo_L, \oo_R})$ is $\P$-restricted, and
		\item $(X \cap K) \subseteq f^{-1}(\LR{\oo_L, \oo_R})$, and if $t_1 \in X$, then $f(t_1) = \oo_L$.
	\end{itemize}
\end{definition}

The algorithm will be recursive, and it will compute a multivariate polynomial in the variables $Z_W$ and $Z_X$, where the coefficient of the term $Z_W^w Z_X^i$ is equal to the cardinality of $\hat{Q}_w^i(V') \coloneqq \LR{(X, C) \in \hat{Q}(V') : \w(X) = w, |X| = i}$, modulo $2$. That is, the formal variables will keep track of the weight and the size of the solutions. The polynomial is computed by using a recursive algorithm that uses the weighted treedepth decomposition to guide recursion. The algorithm starts at the root $r$ and proceeds towards the leaves.

Recall that each node $u_i \in V(F)$ is bijectively mapped to a $V_i \in \P$. The algorithm will assign a value to every vertex $v \in V_i$ from the following set $\states = \{\oo_L, \oo_R, \zz\}$, with the condition that if $v \in K \cap V_i$, then it cannot be assigned $\zz$. The interpretation of the states $\oo_L$ and $\oo_R$ for a vertex $v \in V_i$ is that $v$ is part of a candidate Steiner Tree solution, and is part of the left and the right side of the consistent cut, respectively. On the other hand, the vertices that are not part of a candidate Steiner Tree solution have the state $\zz$. 

Consider a node $u_i \in V(F)$, and a $\P$-restricted function $f: \tail[u_i] \to \states$, we define the set of partial solutions at $u_i$, but excluding any subset of $V_i$, that respect $f$ by
\begin{align}
	\C_{(u_i)}(f) \coloneqq \Big\{(X, \lrcut) \in \hat{Q}(\tree(u_i)) : &X' = X \cup f^{-1}(\LR{\oo_L, \oo_R}),\nonumber 
	\\& C' = (X_L \cup f^{-1}(\oo_L), X_R \cup f^{-1}(\oo_R)), \nonumber
	\\& (X', C') \in \hat{Q}(\broom[u_i])	\Big\} \label{fv:eqn:steiner-cuts-excl}
\end{align}

That is, the partial solutions in $\C_{(u_i)}(f)$ are given by consistently cut subgraphs of $G[\tree(u_i)]$, that are extended to the candidate-cut-pairs for $G[\broom[u_i]]$ by $f$, i.e., consistently cut subgraphs of $G[\broom[u_i]]$ that contain all terminals in $\broom[u_i]$. 

Similarly, for a node $u_i \in V(F)$, and a $\P$-restricted function $g: \tail(u_i) \to \states$, we define the set of partial solutions at $u_i$, but possibly including a subset of $V_i$, that respect $g$ by
\iflong
\begin{align}
	\C_{[u_i]}(g) \coloneqq \Big\{(X, \lrcut) \in \hat{Q}(\tree[u_i]) : &X' = X \cup g^{-1}(\LR{\oo_L, \oo_R}), \nonumber
	\\& C' = (X_L \cup g^{-1}(\oo_L), X_R \cup g^{-1}(\oo_R)), \nonumber
	\\& (X', C') \in \hat{Q}(\broom[u_i])	\Big\} \label{fv:eqn:steiner-cuts-incl}
\end{align}
\fi

With these definitions, the coefficients of the terms $Z_W^w Z_X^k$, for $0 \le w \le 2n^2$ in the polynomial $P_{[r]}(\emptyset)$ at the root node $r \in V(F)$ will give the desired quantities. 

\paragraph{Recursively Computing Polynomials.}
First we define how to compute the polynomials using recurrences. Then, we prove the correctness.

Let $u_i \in V(F)$, and let $f: \tail[u] \to \states$ be a $\P$-restricted function. If $u_i$ is a leaf in $F$, then 
\begin{align}
	P_{(u_i)}(f) &= \left[ (f^{-1}(\oo_L), f^{-1}(\oo_R)) \text{ is a consistent cut of } G[ f^{-1}( \LR{\oo_L, \oo_R} ) ] \right]  \nonumber
	\\&\cdot \ \ \left[ K \cap \tail[u_i] \subseteq f^{-1}(\LR{ \oo_L, \oo_R }) \right] \cdot \left[ t_1 \in \tail[u_i] \implies f(t_1) = \oo_L \right] \label{fv:eqn:steiner-excl-leaf}
\end{align}
If $u_i \in V(F)$ is not a leaf, then 
\begin{align}
	P_{(u_i)}(f) = \prod_{u_j \in \child(u_i)} P_{[u_j]}(f) \label{fv:eqn:steiner-excl-internal}
\end{align}
To define the computation of $P_{[u_i]}(g)$ for a $\P$-restricted function $g: \tail(u_i) \to \states$, we need the following notation. Let $\F(V_i)$ be a set of $\P$-restricted functions (see Definition \ref{fv:defn:p-restricted}) from $V_i \to \states$ that also satisfy the following additional property: for all $h \in \F(V_i)$, if $u, v \in h^{-1}(\LR{\oo_L, \oo_R})$ with $uv \in E(G)$, then $h(u) = h(v)$. We refer to this additional property as the function being \emph{cut-respecting}.

Note $g$ and any $h \in \F(V_i)$ have disjoint domains, and both are $\P$-restricted. Therefore, $g \oplus h$ is also $\P$-restricted for any $h \in \F(V_i)$. We have the following recurrence:
\begin{align}
	P_{[u_i]}(g) = \sum_{h \in \F(V_i)} P_{(u_i)}(g \oplus h) \cdot Z_W^{\w(V_i(h, \oo))} Z_X^{|V_i(h, \oo)|} \label{fv:eqn:steiner-incl}
\end{align}
Where, we use the shorthand $V_i(h, \oo)$ for the set $V_i(h, \oo_L) \cup V_i(h, \oo_R)$.

In fact, we can replace $\F(V_i)$ with its superset $\F'(V_i)$ of $\P$-restricted functions, but without the additional requirement that the function be cut-respecting (see above). We want to claim that one can replace the sum over $\F(V_i)$ with a sum over $\F'(V_i)$ in (\ref{fv:eqn:steiner-incl}), which does not change the result of the computation.

This is because, if we have a function $f: X \to V(G)$ for some $X \subseteq V(G)$, such that $f$ is not cut-respecting, then any polynomial of the form $P_{(u_i)}(f)$ and $P_{[u_i]}(f)$ is equal to $0$ (i.e., the zero polynomial), for any $u_i \in V(F)$. This claim follows easily from observing that at a leaf node (cf. Equation \ref{fv:eqn:steiner-excl-leaf}), we first check whether the function corresponds to a consistent cut, otherwise the corresponding polynomial is defined to be the zero polynomial. The claim follows via a straightforward induction over the definitions (\ref{fv:eqn:steiner-excl-leaf}-\ref{fv:eqn:steiner-incl}). Furthermore, any ``extension'' of a function that is \emph{not} cut-respecting, is also not cut-respecting. Then, loosely speaking, a polynomial at an internal node corresponding to such a function is ``zeroed-out'' due to the recursive definitions. We omit the formal proof that turns this argument into a proof by induction. 

\paragraph{Correctness.}
Now we prove the correctness of (\ref{fv:eqn:steiner-excl-leaf}-\ref{fv:eqn:steiner-incl}). That is, we want to prove that for any $u_i \in V(F)$, and $\P$-consistent functions $f: \tail[u_i] \to \states$, $g: \tail(v_i) \to \states$, and any $w \in [2n^2], i \in [k']$,
\begin{enumerate}
	\item The coefficient of $Z_W^w Z_X^i$ in $P_{(u_i)}(f)$ is equal to 
	$|\LR{(X, C) \in \C_{(u_i)}(f) : \w(X) = w \text{ and } |X| = i }|\mod 2$, and 
	\item The coefficient of $Z_W^w Z_X^i$ in $P_{[u_i]}(g)$ is equal to 
	$|\LR{(X, C) \in \C_{[u_i]}(g) : \w(X) = w \text{ and } |X| = i }|\mod 2$
\end{enumerate}

The proof is by induction.
\paragraph{Base Case.}
First consider the case where $u_i$ is a leaf of $F$, and fix functions $f$ and $g$ as above. Then, note that $\tree(u_i)$ is empty. Therefore, from the definition (\ref{fv:eqn:steiner-cuts-excl}), $\C_{(u_i)}(f)$ is equal to the singleton set $\{(\emptyset, (\emptyset, \emptyset))\}$, if all three predicates in Equation \ref{fv:eqn:steiner-excl-leaf} are true. Otherwise, $\C_{(u_i)}(f)$ is empty. 

Now, consider Equation \ref{fv:eqn:steiner-incl} at a leaf $u_i$ for some $g: \tail(u_i) \to \states$.
Recall that $\tree[u_i] = V_i$, and $\hat{Q}(V_i) = \{ (X, \lrcut) \in \C(V_i) : X \in \hat\R(V_i) \text{ and } t \in V_i \implies t_1 \in X_L \}$. Note that there is a one-to-one correspondence between $(X, \lrcut) \in \hat{Q}(V_i)$ and a function $h \in \F(V_i)$. Now, consider a $(X, \lrcut) \in \hat{Q}(V_i)$ such that $(X', C')$ as in the definition \ref{fv:eqn:steiner-cuts-incl} of $\C_{[u_i]}(g)$, belongs to $\hat{Q}(\broom[u_i])$. 

Then, $g \oplus h$ is a $\P$-restricted function from $\tail[u_i] \to \states$, which implies that $P_{(u_i)}(g \oplus h)$ is equal to $1$. Furthermore, note that $V_i(h, \oo) = X$. Therefore, by multiplying $P_{(u_i)}(g \oplus h)$ with the monomial $Z^{\w(X)}_W Z_X^{|X|}$, we get the term corresponding to $h \in \F'(v_i)$ in Equation \ref{fv:eqn:steiner-incl}. If, on the other hand, $h$ corresponds to a $(X, \lrcut)$ such that $(X', C')$ does not belong to $\hat{Q}(\broom[u_i])$, then $P_{(u_i)}(g \oplus h)$ is equal to $0$. Then, summing over all $h \in \F'(V_i)$, we prove the second property for the base case.
\paragraph{Inductive Hypothesis.}
Now let us inductively assume that the claim is true for all children $u_j \in \child(u_i)$. We want to prove that the same holds for $u_i$.
\paragraph{Inductive Step.}
Consider Equation \ref{fv:eqn:steiner-excl-internal}, for a fixed $\P$-restricted function $f: \tail[u_i] \to \states$. Consider a $(X, \lrcut) \in \hat{Q}(\tree(u_i))$ as in (\ref{fv:eqn:steiner-cuts-excl}). Note that $X \subseteq \tree(u_i)$. For every $u_j \in \child(u_i)$, define $X^j \coloneqq X_j \cap \tree[u_j]$, and the sets $X^j_L, X^j_R$ are defined analogously. Then, since $F$ is a treedepth decomposition of $G$, there are no edges of $E(G)$ between vertices in different $\tree[u_j]$'s. It follows that $(X^j, (X^j_L, X^j_R))$ is a consistently cut subgraph of $\tree[u_j]$. Finally, observe that for any $u_j \in \child(u_i)$, $f$ is a function from $\tail(u_j) \to \states$. Therefore, the inductive hypothesis applies, and we can obtain the size and the weight of $X$ by summing the sizes and the weights of $X^j$'s respectively, over the children $u_j$ of $u_i$. It can be seen that multiplying the respective polynomials $P_{[u_j]}(f)$ accomplishes this.

Consider Equation \ref{fv:eqn:steiner-cuts-excl} for a $\P$-restricted function $g: \tail(u_i) \to \states$. Consider any $(X, \lrcut) \in \C_{[u_i]}(g)$. It follows that $X \in \R(\tree[v_i])$, in particular, $X \subseteq \tree[v_i]$, and $X$ is $\P$-restricted. Let $Y \coloneqq X \cap V_i, Y_L \coloneqq X_L \cap V_i$, and $Y_R \coloneqq X_R \cap V_i$. We then have following properties:
\begin{align*}
	\text{(1)}\ \ |Y \setminus K| \le \k^2 (\Delta+1), \text{\quad and\quad  (2)}\ \ (V_i \cap K) \subseteq Y,\text{ and if }t_1 \in V_i,\text{ then }t_1 \in Y_L.
\end{align*}
Therefore, there is a one-to-one correspondence between $(Y, (Y_L, Y_R))$, and a function $h \in \F'(V_i)$. Now observe that $g \oplus h$ is a $\P$-restricted function from $\tail[u_i]$ to $\states$, which lets us use the correctness of $P_{[u_i]}(g \oplus h)$ from the previous case. Then, using an argument similar to the second case of the base case, the correctness of Equation \ref{fv:eqn:steiner-incl} follows. This completes the proof of correctness by induction.

\iflong
A full description of the procedure \texttt{CountC} is given in Algorithm \ref{fv:alg:steiner-tree}.
\begin{algorithm}
	\caption{\texttt{CountC} for \textsc{Steiner Tree}} \label{fv:alg:steiner-tree}
	\begin{algorithmic}[1]
		\Statex \textbf{Input}: $(G, d, \P, \Gp)$, weighted treedepth decomposition $(F, \varphi)$, $\w: V(G) \to [2n]$, and $w \in [2n^2]$
		\State $P$ = \texttt{calc\_poly\_inc}$(r, \emptyset)$, where $r$ is the root of $F$
		\State \Return the coefficient of $Z_W^w Z_X^k$ in $P$.
		\vspace{0.1cm} \hrule \vspace{0.1cm}		
		\State \textbf{procedure} \texttt{calc\_poly\_exc}($u_i \in V(F), f: \tail[u_i] \to \states$)
		\If {$u_i$ is a leaf in $F$} \Return the polynomial computed using Equation \ref{fv:eqn:steiner-excl-leaf}
		\Else
		\State $P = 1$
		\For{each $u_j \in \child(u_i)$}
		\State $P \gets P \cdot \texttt{calc\_poly\_inc}(u_j, f)$ \Comment{See Equation \ref{fv:eqn:steiner-excl-internal}}
		\EndFor
		\State \Return $P$
		\EndIf
		\State \textbf{end procedure}
		\vspace{0.1cm} \hrule \vspace{0.1cm}		
		\State \textbf{procedure} \texttt{calc\_poly\_inc}($u_i \in V(F), g: \tail(u_i) \to \states$)
		\State $P = 0$
		\For{each $h \in \F(V_i)$}
		\State Compute $y = \w(V_i(h, \oo))$, and $s = |V_i(h, \oo)|$.
		\State $P \gets P + (Z_W^y Z_X^s \cdot$ \texttt{calc\_poly\_inc}$(u_i, g \oplus h))$ \Comment{See Equation \ref{fv:eqn:steiner-incl}}
		\EndFor
		\State \Return $P$
		\State \textbf{end procedure}
	\end{algorithmic}
\end{algorithm}
\fi

Now we prove the following key lemma that bounds the size of each $|\F(V_i)|$. 

\begin{lemma} \label{fv:lem:steiner-function-bound}
	For any $V_i \in \P$, $|\F(V_i)| \le (1+|V_i|)^{\Oh(1)} = 2^{\Oh(\omega(V_i))}$. Furthermore, the set $\F(V_i)$ can be computed in $\text{poly}(|\F(V_i)|, n)$ time.
\end{lemma}
\begin{proof}
	Let $K_i = V_i \cap K$.	Because of the first property from the definition of $\P$-restricted functions, there are at most $(1+|V_i|)^{\k^2(1+\Delta)}$ choices for selecting a subset $U_i \subseteq V_i \setminus K$ of size at most $\k^2(1+\Delta)$, to be mapped to $\{\oo_L, \oo_R\}$. Let us fix such a choice $U_i$. Note that every terminal in $K_i \coloneqq K \cap V_i$ must be assigned to $\{\oo_L, \oo_R\}$.
	
	Recall that $V_i$ is a union of at most $\k$ cliques. Therefore, due to the second property, if there are two vertices $u, v \in U_i \cup K_i$ that belong to the same clique, then they must belong to the same side of the consistent cut. Therefore, there are at most $2^\k$ choices for assigning vertices in $U_i \cup K_i$ to either side of a consistent cut. Therefore, we have the following.
	\begin{align*}
		|\F(V_i)| &\le (1+|V_i|)^{\k^2(1+\Delta)} \cdot 2^\k 
		\\&\le (1+|V_i|)^{\k^2(1+\Delta) + \k} \tag{$\because 1+|V_i| \ge 2$}
		\\&= 2^{\Oh(\omega(V_i))} \tag{Since $\k, \Delta = \Oh(1)$, and $\omega(V_i) = \log(1+|V_i|)$}
	\end{align*}
	Here we would like to highlight the distinction between the weights $\omega: \P \to \real^+$ from the weighted treedepth decomposition, the weights $\w: V(G) \to \natural$ from the Isolation Lemma, and the target weight $w$ for $\w$.
	
	It is relatively straightforward to convert this proof into an algorithm for computing $\F(V_i)$. First, we can use a standard algorithm (e.g., \cite{knuth2005art}) to generate subsets $U_i$ of size at most $\k^2 (1+\Delta)$. It is known that this can be done in $\text{poly}(|V_i|^{\k^2(1+\Delta)})$ time. 
	
	Now, fix a particular choice of $U_i$, and consider the set $U_i \cup K_i$ as defined above. Now we compute an inclusion-wise maximal independent set $S_i$ of $U_i \cup K_i$, e.g., by a greedy algorithm.  Since $V_i$ is a union of at most $\k$ cliques, $|S_i| \le \k$. Now we consider at most $2^\kappa$ choices for assigning $\LR{\oo_L, \oo_R}$ to each vertex in $S_i$. For any vertex $v \in (V_i \cup K_i) \setminus S_i$, there is a vertex $v' \in S_i$ such that $vv'$ is an edge. Therefore, we set $f(v) = f(v')$. Note that if $v$ has more than one neighbor in $S_i$, and if a particular choice assigns them different values, then this corresponds to a function that is \emph{not} cut-respecting. In this case, we may move to the next assignment to $S_i$. Finally, since $S_i$ is a maximal independent set, each cut-respecting function for the fixed choice of $V_i$ will be considered in this manner. Finally, iterating over all choices of $V_i$, we can compute the set $\F(V_i)$ as claimed.
\end{proof}

\begin{theorem} \label{fv:theorem:steiner-tree}
	There exists a $\tndd$ time, polynomial space, randomized algorithm to solve \textsc{Steiner Tree} in the intersection graphs of similarly sized fat objects in $\real^d$.
\end{theorem}
\begin{proof}
	We begin by analyzing the running time and space requirement of the procedure \texttt{CountC} that recursively computes polynomials of interest. From Lemma \ref{fv:lem:steiner-function-bound}, it follows that at every $u_i \in F$, the algorithm spends $2^{\Oh(\omega(u_i))} \cdot n^{O(1)}$ time, and makes $2^{\Oh(\omega(u_i))}$ recursive calls to its children, corresponding to each function in $\F(V_i)$. Furthermore, since the weights defined by $\w: V(G) \to [2n]$ are polynomially bounded, at every node in $F$ the algorithm uses space polynomial in $n$. The correctness of \texttt{CountC} follows from that of the corresponding recurrence relations.
	
	We finally observe that the \cnc algorithm is a randomized procedure makes polynomially many calls to \texttt{CountC}. The correctness and the bounds on probability follow from Corollary \ref{fv:cor:cnc}.
\end{proof}

\iflong
\subsection{Connected Vertex Cover}

\begin{definition}[\textsc{Connected Vertex Cover}]
	\ \\Input: An undirected graph $G = (V, E)$, and an integer $k$. 
	\\Question: Is there a subset $X \subseteq V(G)$ with $|X| \le k'$, such that $X$ is a vertex cover of $G$, and $G[X]$ is connected?
\end{definition}

We will only give the crucial definitions for the rest of the problems that use the \cnc, and explain the important differences from the \textsc{Steiner Tree} algorithm. The details and formal proofs of correctness involve similar ideas, and are therefore omitted.

Fix $(G, d, \P, \Gp)$, and compute a weighted treedepth decomposition $(F, \varphi)$ in $\tndd$ time and polynomial space using Theorem \ref{fv:thm:treedepth}. We have the following observation, that is analogous to Observation \ref{fv:obs:indepset-kpart}.

\begin{observation} \label{fv:obs:cvc}
	Let $X \subseteq V(G)$ be a Connected Vertex Cover of $G$. Then, for any $V_i \in \P$, $|V_i \setminus X| \le \k$.
\end{observation}
This follows from the fact that $V(G) \setminus X$ is an independent set in $G$. Since each $V_i \in \P$ is a union of at most $\k$ cliques, at most $\k$ vertices in each $V_i$ can be in $V(G) \setminus X$.

This motivates the following definition. We say that a set $X \subseteq V(G)$ is $\P$-restricted (in the context of \textsc{Connected Vertex Cover}), if for any $V_i \in \P$, $|V_i \setminus X| \le \k$. 

We guess a vertex $v_1 \in V(G)$ that belongs to a Connected Vertex Cover, and force it to belong to the left side of the cuts. Then, the sets $\R, \S$, and $\Q$ are defined as follows. 
\begin{align*}
	\R &= \LR{X \subseteq V(G) \quad\ : X \text{ is } \P\text{-restricted, is a vertex cover of }G, \text{ and } |X| = k}
	\\\S &= \LR{S \in \R \qquad\quad\ : G[X] \text{ is connected}}
	\\\Q &= \LR{(X, \lrcut): X \in \R \text{ and } v_1 \in X_L}
\end{align*}
The weight restricted versions $\R_w, \S_w$, and $\Q_w$ are defined analogously. The \cnc algorithm for \textsc{Connected Vertex Cover} also uses the same set of $\states \coloneqq \{\oo_L, \oo_R, \zz\}$ as in \textsc{Steiner Tree}. The interpretation of the states $\oo_L, \oo_R$, is now, of course, that a certain vertex is part of a candidate vertex cover, and the subscript denotes the side of the cut to which it belongs.

The definition of a $\P$-restricted function is only slightly modified from Definition \ref{fv:defn:p-restricted} -- in the third condition, the terminal $t_1$ from the Steiner Tree is replaced by the guessed vertex $v_1$ that is forced to be on the left side.

Consider a node $u_i \in V(F)$, and $\P$-restricted functions $f: \tail[u_i] \to \states, g: \tail(u_i) \to \states$. The sets $\C_{(u_i)}(f)$, and $\C_{[u_i]}(g)$ are defined exactly as in (\ref{fv:eqn:steiner-cuts-excl}) and (\ref{fv:eqn:steiner-cuts-incl}) respectively. Now we state how to compute the corresponding polynomials.

If $u_i$ is a leaf in $F$, then for a $\P$-restricted function $f: \tail[u_i] \to \states$, we have the following:
\begin{align}
	P_{(u_i)}(f) &= \left[ (f^{-1}(\oo_L), f^{-1}(\oo_R)) \text{ is a consistent cut of } G[ f^{-1}( \LR{\oo_L, \oo_R} ) ] \right]  \nonumber
	\\&\cdot \ \ \left[ f^{-1}(\LR{\oo_L, \oo_R}) \text{ is a vertex cover of } G[\tail[u_i]] \right] 
	\\&\cdot\ \ \left[ v_1 \in \tail[u_i] \implies f(v_1) = \oo_L \right] \label{fv:eqn:cvc-excl-leaf}
\end{align}
If $u_i \in V(F)$ is not a leaf, then 
\begin{align}
	P_{(u_i)}(f) = \prod_{u_j \in \child(u_i)} P_{[u_j]}(f) \label{fv:eqn:cvc-excl-internal}
\end{align}

As before, let $\F(V_i)$ be the set of $\P$-restricted functions from $V_i$ to $\states$, with the additional property that each function in $\F(V_i)$ be cut-respecting. Then we have the following recurrence:
\begin{align}
	P_{[u_i]}(g) = \sum_{h \in \F(V_i)} P_{(u_i)}(g \oplus h) \cdot Z_W^{\w(V_i(h, \oo))} Z_X^{|V_i(h, \oo)|} \label{fv:eqn:cvc-incl}
\end{align}
As in the previous subsection, we can replace $\F(V_i)$ with the corresponding superset $\F'(V_i)$ without the additional requirement that the function be cut-respecting. As before, this does not affect the result of computation, because, again, the first condition in (\ref{fv:eqn:cvc-excl-leaf}) checks whether the function corresponds to a consistent cut. Then, any polynomial computed using the recurrences above, corresponding a function that is not cut-respecting, will be set to the zero polynomial. Using this, we can prove the correctness of the recurrences (\ref{fv:eqn:cvc-excl-leaf}-\ref{fv:eqn:cvc-incl}). The proof is very similar to the previous section, hence we omit the details.

We prove the following Lemma analogous to Lemma \ref{fv:lem:steiner-function-bound}.

\begin{lemma} \label{fv:lem:cvc-function-bound}
	For any $V_i \in \P$, $|\F(V_i)| \le (1+|V_i|)^{\Oh(1)} = 2^{\Oh(\omega(V_i))}$. Furthermore, the set $\F(V_i)$ can be computed in $\text{poly}(|\F(V_i)|, n)$ time.
\end{lemma}
\begin{proof}
	Because of the definition of $\P$-restricted functions, there are at most $(1+|V_i|)^{\k}$ choices for selecting a subset $U_i \subseteq V_i$ such that $|V_i \setminus U_i| \le \k$. 
	
	Recall that $V_i$ is a union of at most $\k$ cliques. Therefore, due to the second property, if there are two vertices $u, v \in U_i$ that belong to the same clique, then they both must be mapped to $\oo_L$, or both to $\oo_R$. Therefore, there are at most $2^\k$ choices to map vertices in $U_i \cup K_i$ that are part of the same clique to either $\oo_L$, or all of them to $\oo_R$. Therefore, we have the following.
	\begin{align*}
		|\F(V_i)| &\le (1+|V_i|)^{\k} \cdot 2^\k 
		\\&\le (1+|V_i|)^{2\k} \tag{$\because 1+|V_i| \ge 2$} 
		\\&= 2^{\Oh(\omega(V_i))} \tag{Recalling that $\omega(V_i) = \log(1+|V_i|)$, and $\k = \Oh(1)$}
	\end{align*}
	The algorithm for computing the set $\F(V_i)$ in the claimed running time is similar to that in Lemma \ref{fv:lem:steiner-function-bound}, and is therefore omitted.
\end{proof}

One can then use these recurrences to compute the polynomials using a recursive algorithm similar to Algorithm \ref{fv:alg:steiner-tree}. We use the bound on $|\F(u_i)|$ from Lemma \ref{fv:lem:cvc-function-bound}, and appeal to Theorem \ref{fv:thm:subexp-poly}. We omit the details.

\begin{theorem} \label{fv:thm:cvc}
	There exists a $\tndd$ time, polynomial space, randomized algorithm to solve \textsc{Connected Vertex Cover} in the intersection graphs of similarly sized fat objects in $\real^d$.
\end{theorem}

We note the similarity between the various definitions and recurrences for \textsc{Connected Vertex Cover} and that for \textsc{Steiner Tree} from the previous subsection. Indeed, it is shown in \cite{HegerfeldK20} that one can consider \textsc{Connected Vertex Cover} as a special case of \textsc{Steiner Tree} by subdividing edges and making the middle vertex on each subdivided edge a terminal. This only increases the (unweigted) treedepth of the new graph by $1$. However, we cannot use this reduction in our application, since the graph obtained by subdividing edges may not belong to the class of intersection graph of similarly sized fat objects, even when the original graph does.

\subsection{Feedback Vertex Set}

\begin{definition}[\textsc{Feedback Vertex Set}]
	\ \\Input: An undirected graph $G = (V, E)$, and an integer $k$. 
	\\Question: Is there a subset $X \subseteq V(G)$ with $|X| = k$, such that $G-X$ is a forest?
\end{definition}

Fix $(G, d, \P, \Gp)$, and compute a weighted treedepth decomposition $(F, \varphi)$ in $\tndd$ time and polynomial space using Theorem \ref{fv:thm:treedepth}. We have the following observation, that is analogous to Observation \ref{fv:obs:indepset-kpart}.

\begin{observation} \label{fv:obs:fvs}
	Let $X \subseteq V(G)$ be a Feedback Vertex Set of $G$. Then, for any $V_i \in \P$, $|V_i \setminus X| \le 2\k$.
\end{observation}
Note that from any clique $C$ in $G$, there can be at most $2$ vertices that do not belong to a Feedback Vertex Set. Since each $V_i \in \P$ is a union of at most $\k$ cliques, at most $2\k$ vertices in each $V_i$ can be in $V(G) \setminus X$.

The \cnc based algorithm in \cite{cygan2011solving,HegerfeldK20} for \textsc{Feedback Vertex Set} is slightly different from the previous two algorithms, due to the fact that the \textsc{Feedback Vertex Set} problem has a negative connectivity requirement. The idea is to use pairs $(Y, M) \subseteq V(G) \times V(G)$, where $Y$ is the forest that remains after removing an FVS $X \subseteq V(G)$, and $M \subseteq Y$ is a set of \emph{marker vertices} that help in counting the number of connected components. We say that a set $Y \subseteq V(G)$ is $\P$-restricted if $|Y \cap V_i| \le 2\kappa$ for any $V_i \in \P$. The sets $\R, \S$, and $\Q$ are defined as follows.
\begin{align*}
	\R &= \LR{(Y, M) : M \subseteq Y \subseteq V(G), Y \text{ is $\P$-restricted, and } |Y| = n-k },
	\\\S &= \LR{(Y, M) \in \R : G[Y] \text{ is a forest, and every connected component of $G[Y]$ intersects $M$}}
	\\\Q &= \LR{((Y, M), (Y_L, Y_R)) : (Y, M) \in \R, (Y, Y_L, Y_R) \in \C(V), \text{ and } M \subseteq Y_L}
\end{align*}
We need to use larger universe $U$ for defining the weights corresponding to the Isolation Lemma (Lemma \ref{fv:lem:isolation}). Define $U = V \times \{\mathbf{F, M}\}$, and the weight function $\w: U \to [N]$ assigns two different weights $\w(v, \mathbf{F}), \w(v, \mathbf{M})$ depending on whether $v$ is marked or not. To be able to use Corollary \ref{fv:cor:cnc}, we associate $(Y, M)$ with the set $(Y \times \{\mathbf{F}\}) \cup (M \times \LR{\mathbf{M}}) \subseteq U$. Then, we can also define $\w(Y, M) = \w(Y \times \{\mathbf{F}\}) \cup (M \times \LR{\mathbf{M}}).$ However, for conceptual ease, we will describe the algorithm using the earlier notation.

\begin{lemma}[\cite{cygan2011solving}]
	Let $(Y, M)$ be such that $M \subseteq Y \subseteq V(G)$. The number of consistently cut subgraphs $(Y, (Y_L, Y_R))$ such that $M \subseteq Y_L$ is equal to $2^{\overline{\cc}_M(G[Y])}$, where $\overline{\cc}_M(G[Y])$ us the number of connected components of $G[Y]$ that do not contain any vertex from $M$.
\end{lemma}
The proof of this lemma is similar to that of Lemma \ref{fv:lem:num-ccs}, and can be found in \cite{cygan2011solving,HegerfeldK20}. We define further subsets of $\R, \S$, and $\Q$ based on the number of edges, markers and weight.
\begin{align*}
	\R_{w}^{j, \ell} &= \{ (Y, M) \in \R &: \w(Y, M) = w, |E(G[Y])| = j, |M| = \ell \}
	\\\S_{w}^{j, \ell} &= \{ (Y, M) \in \S &: \w(Y, M) = w, |E(G[Y])| = j, |M| = \ell \}
	\\\Q_{w}^{j, \ell} &= \{ ((Y, M), (Y_L, Y_R)) \in \Q &: \w(Y, M) = w, |E(G[Y])| = j, |M| = \ell \}
\end{align*}
We have the following lemma analogous to Lemma \ref{fv:lem:num-ccs} from \cite{cygan2011solving}, the proof of which is omitted.
\begin{lemma}[\cite{cygan2011solving,HegerfeldK20}] \label{fv:lem:qs-fvs}
	Let $\w: U \to [N]$ be a weight function. Then for every $w \in \natural$, and $j \in [n-k-1]$, we have that $|\S_w^{j, n-k-j}| \equiv |\Q_w^{j, n-k-j}|$.
\end{lemma}

From this lemma, it follows that $G$ has a Feedback Vertex Set $X$ of size $k$ iff for some $w, j \in \natural$ and $M \subseteq Y \coloneqq (V(G) \setminus X)$ such that $(Y, M) \in \S_w^{j, n-k-j}$. Now the \texttt{CountC} algorithm for \textsc{Feedback Vertex Set} will accomplish this task using polynomial computation. Define
\begin{align*}
	\hat\R(V') &= \LR{(Y, M) : M \subseteq Y \subseteq V' \text{ and } Y \text{ is $\P$-restricted} }
	\\\hat\Q(V') &= \LR{ ((Y, M), (Y_L, Y_R)) : (Y, M) \in \hat\R(V'), (Y, (Y_L, Y_R)) \in \C(V'), \text{ and } M \subseteq Y_L  }
\end{align*}
The multivariate polynomials will have four formal variables $Z_W, Z_Y, Z_E, Z_M$. The set of $\states$ is defined to be $\{\oo, \zz_L, \zz_R\}$, where $\oo$ represents that the vertex is in a Feedback Vertex Set; whereas the states $\zz_L, \zz_R$ represent that the vertex is in the remaining forest, and the subscript denotes the side of the consistent cut the vertex belongs to. We modify the definition of $\P$-restricted functions as follows.
\begin{definition}
	Let $f: X \to \states$ be a function, where $X \subseteq V(G)$. We say that $f$ is $\P$-restricted, if $f^{-1}(\LR{\zz_L, \zz_R})$ is $\P$-restricted.
\end{definition}

Now, for a node $u_i \in V(F)$, and a $\P$-restricted function $f: \tail[u_i] \to \states$, we define partial solutions at $u_i$, but excluding $V_i$, that respect $f$ by:
\begin{align}
	\C_{(u_i)}(f) \coloneqq \Big\{ ((Y, M), (Y_L, Y_R)) \in \hat{Q}(\tree(u_i)) : &Y' = Y \cup f^{-1}(\LR{\zz_L, \zz_R}),\nonumber 
	\\& C' = (Y_L \cup f^{-1}(\zz_L), Y_R \cup f^{-1}(\zz_R)), \nonumber
	\\& ((Y', M), C') \in \hat{Q}(\broom[u_i])	\Big\} \label{fv:eqn:fvs-cuts-excl}
\end{align}
Similarly, for a $\P$-restricted function $g: \tail(u_i) \to \states$, we define partial solutions at $u_i$, but possibly including $V_i$, that respect $g$ by:
\begin{align}
	\C_{[u_i]}(f) \coloneqq \Big\{ ((Y, M), (Y_L, Y_R)) \in \hat{Q}(\tree[u_i]) : &Y' = Y \cup f^{-1}(\LR{\zz_L, \zz_R}),\nonumber 
	\\& C' = (Y_L \cup f^{-1}(\zz_L), Y_R \cup f^{-1}(\zz_R)), \nonumber
	\\& ((Y', M), C') \in \hat{Q}(\broom[u_i])	\Big\} \label{fv:eqn:fvs-cuts-incl}
\end{align}

Now, at a node $u_i \in V(F)$, and corresponding to $\P$-restricted functions $f: \tail[u_i] \to \states, g: \tail(u_i) \to \states$, we will compute polynomials $P_{(u_i)}(f)$ and $P_{[u_i]}(g)$ respectively. These have the following interpretation. The coefficient of the monomial $Z_W^w Z_Y^a Z_E^b Z_M^c$ in $P_{(u_i)}(f)$ is given by
$$\left| \LR{ ((Y, M), (Y_L, Y_R)) \in \C_{(u_i)}(f) :  \w(Y, M) = w, |Y| = a, |E(G[Y])| + |E(Y, f^{-1}(\zz_L, \zz_R)\setminus Y)| = b, |M|= c  } \right| \mod 2.$$
Similarly, the coefficient of the monomial $Z_W^w Z_Y^a Z_E^b Z_M^c$ in $P_{[u_i]}(g)$ is given by
$$\left| \LR{ ((Y, M), (Y_L, Y_R)) \in \C_{[u_i]}(g) :  \w(Y, M) = w, |Y| = a, |E(G[Y])| + |E(Y, g^{-1}(\zz_L, \zz_R))| = b, |M|= c  } \right| \mod 2.$$
Where we are using the notation $E(A, B) \subseteq E(G)$, where $A, B \subseteq V(G)$ are disjoint, to denote the set of edges with one endpoint in $A$ and another in $B$.

Then, it can be seen that the coefficients of the polynomial $P_{[r]}(\emptyset)$ contain the size of the desired sets.

\paragraph{Recursively Computing Polynomials}
Let $u_i \in V(F)$ be a leaf in $F$ and let $f: \tail[u_i] \to \states$ be a $\P$-restricted function. Then,
\begin{align}
	P_{(v)}(f) = \left[ (f^{-1}(\zz_L), f^{-1}(\zz_R)) \text{ is a consistent cut of } G[f^{-1}(\LR{\zz_L, \zz_R})] \right] \label{fv:eqn:fvs-excl-leaf}
\end{align}
If $u_i \in V(F)$ is an internal node, then
\begin{align}
	P_{(u_i)}(f) = \prod_{u_j \in \child(u_i)} P_{[u_j]}(f) \label{fv:eqn:fvs-excl-internal}
\end{align}
Finally, for any $u_i \in V(F)$, we define $\F(V_i)$ to be the $\P$-restricted functions from $V_i$ to $\states$, with the additional requirement that they be cut-consistent. Fix a function $h \in \F(V_i)$, recall that $V_i(h, \oo), V_i(h, \zz_L), V_i(h, \zz_R)$ are the subsets of $V_i$ that are mapped to the corresponding $\texttt{state}$ by $h$. Let $V_i(h, \zz) = V_i(h, \zz_L) \cup V_i(h, \zz_R)$. Let $\w_{\mathbf{F}} = \sum_{v \in V_i(h, \zz)} \w(v, \mathbf{F})$, and $\w_{\mathbf{M}} = \sum_{v \in V_i(h, \zz_L)} \w(v, \mathbf{M})$. Finally, let $j = \lr{\sum_{v \in V_i(h, \zz)} |N(v) \cap g^{-1}(\LR{\zz_L, \zz_R})|} + |E(G[V_i(h, \zz)])|$. Then for any $\P$-consistent function $g: \tail(u_i) \to \states$,
\begin{align}
	P_{[u_i]}(g) = \sum_{h \in \F(V_i)} P_{(u_i)}(g \oplus h) \cdot Z_W^{\w_{\mathbf{F}}}\ Z_Y^{|V_i(h, \zz)|}\ Z_E^{j} \cdot \lr{\prod_{v \in V_i(h, \zz_L)} \lr{1 + Z_W^{\w(v, \textbf{M})} Z_M}} \label{fv:eqn:fvs-incl}
\end{align}
We remark that each of the sub-polynomials corresponding to a function $h \in \F(V_i)$ can be computed in polynomial time and space. Therefore, the $P_{[u_i]}(g)$ can be computed in time proportional to $|\F(V_i)| \cdot n^{O(1)}$.

\paragraph{Correctness}

\paragraph{Base Case.}
Suppose $u_i$ is a leaf of $F$, and fix functions $f$ and $g$ as above. Then, note that $\tree(u_i)$ is empty. Therefore, from the definition of $\C_{(u_i)}(f)$, \ref{fv:eqn:fvs-cuts-excl} is equal to the singleton set $\{((\emptyset, \emptyset), (\emptyset, \emptyset))\}$, if the predicate in \ref{fv:eqn:fvs-excl-leaf} is true. Otherwise, $\C_{(u_i)}(f)$ is empty. 

Now, consider Equation \ref{fv:eqn:fvs-incl} at a leaf $u_i$ for some $g: \tail(u_i) \to \states$.
Recall that $\tree[u_i] = V_i$, and $\hat{Q}(V_i) = \{ ((Y, M), (Y_L, Y_R)) : (Y, M) \in \hat\R(V_i), (Y, (Y_L, Y_R)) \in \C(V_i), \text{ and } M \subseteq Y_L \}$. Now, fix a consistently cut subgraph $(Y, (Y_L, Y_R))$ such that there exists $M \subseteq Y_L$ such that $((Y', M), C')$ as in the definition \ref{fv:eqn:fvs-cuts-incl} of $\C_{[u_i]}(g)$, belongs to $\hat{Q}(\broom[u_i])$. Let $(W, W_L, W_R)$ denote the intersection of the corresponding sets in $(Y, Y_L, Y_R)$ with $V_i$. Note that there is a one-to-one correspondence between a function $h \in \F(V_i)$, and a consistently cut subgraph $(W, (W_L, W_R))$.

Using the correctness of the previous case, we know that the coefficients of the respective terms in the polynomial $P_{(u_i)}(g)$ correspond to the size of the respective sets, modulo $2$. Now we discuss the correctness of the multiplying term in \ref{fv:eqn:fvs-incl}. Note that the weight and size of $Y$ are obtained by adding the respective quantities of $V_i(h, \zz)$ to that of $Y\setminus V_i(h, \zz)$, and the quantities of $Y\setminus V_i(h, \zz)$ are encoded in the polynomial via the inductive hypothesis. 

Note that for every $v \in V_i(h, \zz_L)$, we have two choices -- whether to add a vertex to the set of marked vertices, or not. We account for the two choices by multiplying by the binomial $1 + Z_W^{\w(v, \mathbf{M})} Z_M$. If $v$ is unmarked, then we have already accounted for its weight in the previous paragraph, which corresponds to the term $1$ in the binomial. If, however, $v \in M$, the size of marked vertices increases by $1$, and the weight increases by $\w(v, \mathbf{M})$, which we account in the second term of the binomial.

Finally, we need to update the number of edges, encoded by the formal variable $Z_E$. When we add a vertex $v \in V_i(h, \zz)$ to the set $Y$, the number of edges between $v$ and $g^{-1}(\LR{\zz_L, \zz_R})$ is exactly $|N(v) \cap g^{-1}(\LR{\zz_L, \zz_R})|$. Finally, we need to account for the edges among the vertices in $V_i \cap Y$, which is exactly the quantity $|E(G[V_i(h, \zz)])|$. Note that the number of edges in $N(v) \cap \tree(u_i) \cap Y$ are already accounted for, using the correctness of the previous case with respect to the function $g \oplus f$. This completes the proof of the base case.

\paragraph{Inductive Hypothesis.}
Now let us inductively assume that the claim is true for all children $u_j \in \child(u_i)$. We want to prove that the same holds for $u_i$.

\paragraph{Inductive Step.}
Consider Equation \ref{fv:eqn:fvs-excl-internal}, for a fixed $\P$-restricted function $f: \tail[u_i] \to \states$. Consider a $((Y, M), (Y_L, Y_R)) \in \hat{Q}(\tree(u_i))$ as in (\ref{fv:eqn:fvs-cuts-excl}). Note that $Y \subseteq \tree(u_i)$. For every $u_j \in \child(u_i)$, define $Y^j \coloneqq Y \cap \tree[u_j]$, and the sets $M^j, Y^j_L, Y^j_R$ are defined analogously. 

By the inductive hypothesis, the coefficient of the term $Z_W^w Z_Y^a Z_E^b Z_M^c$ in polynomials $P_{[u_j]}(f)$ for $u_j \in \child(u_i)$ correctly counts the number $(Y^j, M^j), (Y^j_L, Y^j_R) \in \C_{[u_j]}(f)$ with the appropriate counts. By the properties of treedepth decomposition, there are no edges between vertices in different $\tree[u_j]$'s. Therefore, the size and the weight of $Y$, and the size of $M$ can be obtained by adding the respective quantities in the corresponding subtrees. Similarly, again by using the properties of treedepth decomposition, there are no edges between vertices in different $Y^j$'s, which implies that $|E(G[Y])| + |E(Y, f^{-1}(\LR{\zz_L, \zz_R}))| = \sum_{u_j \in \child(u_i)} |E(G[Y^j])| + |E(Y^j, f^{-1}(\LR{\zz_L, \zz_R}))|$. Therefore, the polynomial $P_{(u_i)}(f)$ can be obtained by polynomial multiplication, which corresponds to adding the respective quantities in the powers of the formal variables.

Consider Equation \ref{fv:eqn:fvs-cuts-excl} for a $\P$-restricted function $g: \tail(u_i) \to \states$. Consider any $((Y, M), (Y_L, Y_R)) \in \C_{[u_i]}(g)$. It follows that $Y \in \R(\tree[v_i])$, in particular, $Y \subseteq \tree[v_i]$, and $Y$ is $\P$-restricted. Let $(W, N, W_L, W_R)$ denote the intersection of the respective sets from $(Y, M, Y_L, Y_R)$ with the set $V_i$. We then have following properties:
\begin{enumerate}
	\item $|Y_L \cup Y_R| \le 2\k$,
	\item For any $u \in Y_L$ and $v \in Y_R$, $uv \not\in E(G)$, and
	\item If $t_1 \in V_i$, then $t_1 \in Y_L$.
\end{enumerate}
Therefore, there is a one-to-one correspondence between $((W, N), (W_L, W_R))$, and a function $h \in \F'(V_i)$. Now we observe that $g \oplus h$ is a function from $\tail[u_i]$ to $\states$. Then, using an argument similar to the second part of the base case, the correctness of equation \ref{fv:eqn:fvs-incl} follows. This finishes the proof of correctness of the reucrrences by induction.

The proof of the following lemma follows from Observation \ref{fv:obs:fvs}, is very similar to that of Lemma \ref{fv:lem:cvc-function-bound}, and is therefore omitted.

\begin{lemma} \label{fv:lem:fvs-function-bound}
	For any $V_i \in \P$, $|\F(V_i)| \le (1+|V_i|)^{\Oh(1)} = 2^{\Oh(\omega(V_i))}$. Furthermore, the set $\F(V_i)$ can be computed in time $\text{poly}(|\F(V_i)|, n)$.
\end{lemma}

Then, we get the following result using similar arguments as in the previous sections.
\begin{theorem} \label{fv:theorem:fvs}
	There exists a $\tndd$ time, polynomial space, randomized algorithm to solve \textsc{Feedback Vertex Set} in the intersection graphs of similarly sized fat objects in $\real^d$.
\end{theorem}

\subsection{Connected Odd Cycle Transversal}

\begin{definition}[\textsc{Connected Odd Cycle Transversal}]
	\ \\Input: An undirected graph $G = (V, E)$, and an integer $k$. 
	\\Question: Is there a subset $X \subseteq V(G)$ with $|X| = k$, such that $G[X]$ is connected, and $G-X$ is bipartite?
\end{definition}

Fix $(G, d, \P, \Gp)$, and compute a weighted treedepth decomposition $(F, \varphi)$ in $\tndd$ time and polynomial space using Theorem \ref{fv:thm:treedepth}. We have the following observation.

\begin{observation} \label{fv:obs:coct}
	Let $X \subseteq V(G)$ be a (Connected) Odd Cycle Transversal of $G$. Then, for any $V_i \in \P$, $|V_i \setminus X| \le 2\k$.
\end{observation}
Note that from any clique $C$ in $G$, there can be at most $2$ vertices that do not belong to an Odd Cycle Tranversal -- otherwise there is a triangle in $C$. Since each $V_i \in \P$ is a union of at most $\k$ cliques, at most $2\k$ vertices in each $V_i$ can be in $V(G) \setminus X$. We say that a set $X \subseteq V(G)$ is $\P$-restricted, if $|(V(G) \setminus X) \cap V_i| \le 2\k$ for any $V_i \in \P$.

First we adopt some notation from \cite{HegerfeldK20}. We say that $(A, B)$ is a bipartition of $G$, denoted by $(A, B) \in \bip(G)$, if $\{A, B\}$ is a partition of $V$, and $E(G[A]) = E(G[B]) = \emptyset$. For a set $X \subseteq V(G)$, the counting algorithm will count the number of candidates $(X, A)$, where $X$ is an Odd Cycle Transversal, OCT for short; and $A$ is one side of the bipartition. Formally, 
\begin{align*}
	\R &= \LR{(X, A) \in 2^{V(G)} \times 2^{V(G)}: X \text{ is $\P$-restricted, } X \cap A = \emptyset, (A, V(G) \setminus (X \cup A)) \in \bip(G - X), |X| = k },
	\\\S &= \LR{(X, A) \in \R : G[X] \text{ is connected}}
	\\\Q &= \LR{((X, A), (\lrcut)) : (X, A) \in \R, (X, \lrcut) \in \C(V), \text{ and } v_1 \in X_L}
\end{align*}
Where, we guess a vertex $v_1$ in an OCT solution, and force it on the left side of the consistent cuts, in order to break the left-right symmetry.

Similar to the previous section, we need to define a larger universe $U$ for applying the Isolation Lemma. Define $U = V \times \{\mathbf{X, A}\}$. We identify a pair $(X, A)$ with $X \times \{\mathbf{X}\} \cup A \times \{\mathbf{A}\}$. Then, we can also define $\w(X, A) = \w(X \times \{\mathbf{X}\} \cup A \times \{\mathbf{A}\}).$ However, for conceptual ease, we will describe the algorithm using the earlier notation. For the sets $\R, \S, \Q$, we also define their weight-restricted versions $\R_w, \S_w, \Q_w$, where the weight $\w(X, A) = \w(X \times \LR{\mathbf{X}} \cup A \times \LR{\mathbf{A}})$ in the corresponding definitions is required to be exactly $w$ from the subscript.

We have the following lemma analogous to Lemmata \ref{fv:lem:num-ccs},\ref{fv:lem:qs-fvs} from \cite{cygan2011solving}, the proof of which is omitted.
\begin{lemma}[\cite{cygan2011solving,HegerfeldK20}]
	Let $\w: U \to [N]$ be a weight function. Then for every $w \in \natural$, $|\S_w| \equiv |\Q_w|$.
\end{lemma}

For any subset $V' \subseteq V(G)$
\begin{align*}
	\hat{R}(V') &= \LR{(X, A) \in 2^{V'} \times 2^{V'}: X \text{ is $\P$-restricted, } X \cap A = \emptyset, (A, V' \setminus (X \cup A)) \in \bip(G[V'] - X), |X| = k },
	\\\hat{\Q}(V') &= \LR{((X, A), \lrcut) : (X, A) \in \hat\R(V'), \text{ and } v_1 \in X_L}
\end{align*}

The multivariate polynomials will have two formal variables $Z_W, Z_X$. Define $\states \coloneqq \{\oo_L, \oo_R, \zz_A, \zz_B\}$, where the states $\oo_L, \oo_R$ represent that the vertex is in the (candidate) OCT, and the subscript denotes the side of the consistent cut the vertex belongs to; on the other hand the states $\zz_A, \zz_B$ denotes that a vertex is not in the (candidate) OCT.
\begin{definition}
	Let $f: X \to \states$ be a function, where $X \subseteq V(G)$. We say that $f$ is $\P$-restricted, if $f^{-1}(\LR{\oo_L, \oo_R})$ is $\P$-restricted.
\end{definition}
Now, for a node $u_i \in V(F)$, and a $\P$-restricted function $f: \tail[u_i] \to \states$, we define partial solutions at $u_i$, but excluding $V_i$, that respect $f$ by:
\begin{align}
	\C_{(u_i)}(f) \coloneqq \Big\{ ((X, A), (X_L, X_R)) \in \hat{Q}(\tree(u_i)) : & X' = X \cup f^{-1}(\LR{\oo_L, \oo_R}),\nonumber 
	\\& C' = (Y_L \cup f^{-1}(\oo_L), Y_R \cup f^{-1}(\oo_R)), \nonumber
	\\& A' = A \cup f^{-1}(\zz_A) \nonumber
	\\& ((X', A'), C') \in \hat{Q}(\broom[u_i])	\Big\} \label{fv:eqn:coct-cuts-excl}
\end{align}
Similarly, for a $\P$-restricted function $g: \tail(u_i) \to \states$, we define partial solutions at $u_i$, but possibly including $V_i$, that respect $g$ by:
\begin{align}
	\C_{[u_i]}(g) \coloneqq \Big\{ ((X, A), (X_L, X_R)) \in \hat{Q}(\tree(u_i)) : & X' = X \cup f^{-1}(\LR{\oo_L, \oo_R}),\nonumber 
	\\& C' = (Y_L \cup f^{-1}(\oo_L), Y_R \cup f^{-1}(\oo_R)), \nonumber
	\\& A' = A \cup f^{-1}(\zz_A), \nonumber
	\\& ((X', A'), C') \in \hat{Q}(\broom(u_i))	\Big\} \label{fv:eqn:coct-cuts-incl}
\end{align}

Now, at a node $u_i \in V(F)$, and corresponding to $\P$-restricted functions $f: \tail[u_i] \to \states, g: \tail(u_i) \to \states$, we will compute polynomials $P_{(u_i)}(f)$ and $P_{[u_i]}(g)$ respectively. These have the following interpretation. The coefficient of the monomial $Z_W^w Z_X^i$ in $P_{(v)}(f)$ is given by
$$\left| \LR{ ((X, A), (X_L, X_R)) \in \C_{(u_i)}(f) :  \w(X \times \LR{\mathbf{X}} \cup A \times \LR{\mathbf{A}}) = w, |X| = i  } \right| \mod 2.$$
Similarly, the coefficient of the monomial $Z_W^w Z_X^i$ in $P_{[u_i]}(g)$ is given by
$$\left| \LR{ ((X, A), (X_L, X_R)) \in \C_{[u_i]}(g) :  \w(X \times \LR{\mathbf{X}} \cup A \times \LR{\mathbf{A}}) = w, |X| = i  } \right| \mod 2.$$

\paragraph{Recursively Computing Polynomials}
Let $u_i \in V(F)$ be a leaf in $F$ and let $f: \tail[u_i] \to \states$ be a $\P$-restricted function. Then,
\begin{align}
	P_{(v)}(f) &= \left[ (f^{-1}(\oo_L), f^{-1}(\oo_R)) \text{ is a consistent cut of } G[ f^{-1}( \LR{\oo_L, \oo_R} ) ] \right]  \nonumber
	\\&\cdot \ \ \left[ f^{-1}(\LR{\zz_L, \zz_R}) \text{ is a bipartition } G[f^{-1}(\LR{\zz_A, \zz_B})] \right]  \nonumber
	\\&\cdot\ \ \left[ v_1 \in \tail[u_i] \implies f(v_1) = \oo_L \right] \label{fv:eqn:coct-excl-leaf}
\end{align}
If $u_i \in V(F)$ is an internal node, then
\begin{align}
	P_{(u_i)}(f) = \prod_{u_j \in \child(u_i)} P_{[u_j]}(f) \label{fv:eqn:coct-excl-internal}
\end{align}
Finally, for any $u_i \in V(F)$, we define $\F(V_i)$ to be the $\P$-consistent functions from $V_i$ to $\states$, with the additional requirement that they be cut-respecting. Fix a function $h \in \F(V_i)$, recall that $V_i(h, \texttt{state})$ is the subset of $V_i$ that is mapped to $\texttt{state} \in \states$. Define: $V_i(h, \oo) \coloneqq V_i(h, \oo_L) \cup V_i(h, \oo_R)$,  $\w_{\mathbf{X}} \coloneqq \sum_{v \in V_i(h, \oo)} \w(v, \mathbf{X})$, $\w_{\mathbf{A}} \coloneqq \sum_{v \in V_i(h, \zz_L)} \w(v, \mathbf{A})$, and finally $s \coloneqq |V_i(h, \oo)|$. 

Then for any $\P$-consistent function $g: \tail(u_i) \to \states$,
\begin{align}
	P_{[u_i]}(g) = \sum_{h \in \F(V_i)} P_{(u_i)}(g \oplus h) \cdot Z_W^{\w_{\mathbf{X}} + \w_{\mathbf{A}}} Z_X^{s} \label{fv:eqn:coct-incl}
\end{align}
Similar to the discussion in the previous subsections, it is possible to replace $\F(V_i)$ with its superset $\F'(V_i)$, where we do not have the additional requirement of cut-respecting. Then, due to the first condition in \ref{fv:eqn:coct-excl-leaf}, all polynomials with respect to such a function will be equal to the zero polynomials. Furthermore, it is possible that some functions in $\F(V_i)$ will not induce a bipartition, which will also result in all corresponding polynomials being equal to the zero polynomials. Therefore, it may be possible to further prune the recursive calls by examining each $h \in \F(V_i)$. However, even an upper bound as given in Lemma \ref{fv:lem:coct-function-bound} suffices.

\paragraph{Correctness.}
Now we prove the correctness of recurrences (\ref{fv:eqn:coct-excl-leaf}-\ref{fv:eqn:coct-incl}). That is, we want to prove that for any $u_i \in V(F)$, and $\P$-consistent functions $f: \tail[u_i] \to \states$, $g: \tail(v_i) \to \states$, and any $w \in [2n^2], i \in [k]$,
\begin{enumerate}
	\item Coefficient of $Z_W^w Z_X^i$ in $P_{(u_i)}(f)$ is equal to 
	$$|\LR{((X, A), (X_L, X_R)) \in \C_{(u_i)}(f) : \w(X \times \LR{\mathbf{X}} \cup A \times \LR{\mathbf{A}}) = w \text{ and } |X| = i }|\mod 2$$ 
	\item Coefficient of $Z_W^w Z_X^i$ in $P_{[u_i]}(g)$ is equal to 
	$$|\LR{((X, A), (X_L, X_R)) \in \C_{[u_i]}(g) : \w(X \times \LR{\mathbf{X}} \cup A \times \LR{\mathbf{A}}) = w \text{ and } |X| = i }|\mod 2$$
\end{enumerate}
The proof is by induction.
\paragraph{Base Case.}
First consider the case where $u_i$ is a leaf of $F$, and fix functions $f$ and $g$ as above. Then, note that $\tree(u_i)$ is empty. Therefore, from the definition of $\C_{(u_i)}(f)$, \ref{fv:eqn:coct-cuts-excl} is equal to the singleton set $\{((\emptyset, \emptyset), (\emptyset, \emptyset))\}$, if all three predicates in Equation \ref{fv:eqn:steiner-excl-leaf} are true. Otherwise, $\C_{(u_i)}(f)$ is empty. 

Now, consider Equation \ref{fv:eqn:coct-incl} at a leaf $u_i$ for some $g: \tail(u_i) \to \states$.
Recall that $\tree[u_i] = V_i$, and $\hat{Q}(V_i) = \{ ((X, A), \lrcut) \in \C(V_i) : X \in \hat\R(V_i) \text{ and } t \in V_i \implies t_1 \in X_L \}$. Note that there is a one-to-one correspondence between $(X, \lrcut) \in \hat{Q}(V_i)$ and a function $h \in \F'(V_i)$. Now, consider a $(X, \lrcut) \in \hat{Q}(V_i)$ such that $((X', A'), C')$ as in the definition \ref{fv:eqn:coct-cuts-incl} of $\C_{[u_i]}(g)$, belongs to $\hat{Q}(\broom[u_i])$. 

Then, $g \oplus h$ is a $\P$-restricted function from $\tail[u_i] \to \states$, which implies that $P_{(u_i)}(g \oplus h)$ is equal to $1$. Furthermore, note that $V_i(h, \oo) = X$. Note the total weight and the number of vertices in $V_i(h, \oo)$ is equal to $\w_{\mathbf{X}}$, and $s$ respectively, as defined before (\ref{fv:eqn:coct-incl}). Similarly, the weight of the vertices in $V_i(h, \zz_L)$ is equal to $\w_{\mathbf{A}}$. Therefore, by multiplying $P_{(u_i)}(g \oplus h)$ with the monomial $Z^{\w_{\mathbf{X}}+\w_{\mathbf{A}}}_W Z_X^{s}$, we get the term corresponding to $h \in \F(v_i)$ in Equation \ref{fv:eqn:coct-incl}. If, on the other hand, $h$ corresponds to a $(X, \lrcut)$ such that $(X', C')$ does not belong to $\hat{Q}(\broom[u_i])$, then $P_{(u_i)}(g \oplus h)$ is equal to $0$. Then, summing over all $h \in \F'(V_i)$, we prove the second property for the base case.

\paragraph{Inductive Hypothesis.}
Now let us inductively assume that the claim is true for all children $u_j \in \child(u_i)$. We want to prove that the same holds for $u_i$.

\paragraph{Inductive Step.}
Consider Equation \ref{fv:eqn:coct-excl-internal}, for a fixed $\P$-restricted function $f: \tail[u_i] \to \states$. Consider a $((X, A), \lrcut) \in \hat{Q}(\tree(u_i))$ as in (\ref{fv:eqn:coct-cuts-excl}). Note that $X \subseteq \tree(u_i)$. For every $u_j \in \child(u_i)$, define $X^j \coloneqq X_j \cap \tree[u_j]$, and the sets $A^j, B^j, X^j_L, X^j_R$ are defined analogously. 

Since $F$ is a treedepth decomposition of $G$, there are no edges of $E(G)$ between vertices in different $\tree[u_j]$'s. It follows that $((X^j, A^j), (X^j_L, X^j_R)) \in |\C_{[u_j]}(f)|$. Furthermore, $(A^j, B^j)$ is a bipartition of $G[\tree[u_j]] - X^j$.

Therefore, we can obtain the weight $\w(X, A)$ and $|X|$ by adding the respective quantities $\w(X^j, A^j)$, and $|X^j|$ over the children $u_j \in \child(u_i)$. Then, using the properties of polynomial multiplication, and the correctness of polynomials $P_{[u_j]}(f)$ via the inductive hypothesis, the correctness of Equation \ref{fv:eqn:coct-excl-leaf} follows for an internal node $u_i \in V(F)$.

Consider Equation \ref{fv:eqn:coct-cuts-excl} for a $\P$-restricted function $g: \tail(u_i) \to \states$. Consider any $((X, A), \lrcut) \in \C_{[u_i]}(g)$. It follows that $X \in \R(\tree[v_i])$, in particular, $X \subseteq \tree[v_i]$, and $X$ is $\P$-restricted. Let $(Y, Y_L, Y_R, C, D)$ denote the intersection of the respective sets from $(X, X_L, X_R, A, B)$ with the set $V_i$. We then have the following properties:
\begin{enumerate}
	\item $|C \cup D| \le \k^2 (\Delta+1)$,
	\item For any $u \in Y_L$ and $v \in Y_R$, $uv \not\in E(G)$, 
	\item $E(G[C]) = E(G[D]) = \emptyset$, and
	\item If $t_1 \in V_i$, then $t_1 \in Y_L$.
\end{enumerate}
Therefore, there is a one-to-one correspondence between $((Y, C), (Y_L, Y_R))$, and a function $h \in \F(V_i)$. Now we observe that $g \oplus h$ is a function from $\tail[u_i]$ to $\states$. Then, using an argument similar to the second part of the base case, the correctness of equation \ref{fv:eqn:coct-incl} follows. This finishes the proof of correctness of the reucrrences by induction.

We prove the following bound analogous to Lemma \ref{fv:lem:fvs-function-bound}.

\begin{lemma} \label{fv:lem:coct-function-bound}
	For any $V_i \in \P$, $|\F(V_i)| \le (1+|V_i|)^{\Oh(1)} = 2^{\Oh(\omega(V_i))}$. Furthermore, the set $\F(V_i)$ can be computed in time $\text{poly}(|\F(V_i)|, n)$.
\end{lemma}
\begin{proof}
	Because of the definition of $\P$-restricted functions, there are at most $(1+|V_i|)^{2\k}$ choices for selecting a subset $U_i \subseteq V_i$ such that $|V_i \setminus U_i| \le 2\k$. 
	
	Recall that $V_i$ is a union of at most $\k$ cliques. Therefore, due to the second property, if there are two vertices $u, v \in U_i$ that belong to the same clique, then they both must be mapped to $\oo_L$, or both to $\oo_R$. Therefore, there are at most $2^\k$ choices to map vertices in $U_i \cup K_i$ that are part of the same clique to either $\oo_L$, or all of them to $\oo_R$. Finally, note that there are at most $2\k$ vertices in $V_i \setminus U_i$, each of which are mapped to either $\zz_A$, or to $\zz_B$. 
	Therefore, we have the following.
	\begin{align*}
		|\F(V_i)| &\le (1+|V_i|)^{2\k} \cdot 2^\k \cdot 2^{2\k} 
		\\&\le (1+|V_i|)^{5\k} \tag{$\because 1+|V_i| \ge 2$} 
		\\&= 2^{\Oh(\omega(V_i))} \tag{Recalling that $\omega(V_i) = \log(1+|V_i|)$, and $\k = \Oh(1)$}
	\end{align*}
	The algorithm for computing the set $\F(V_i)$ in the claimed running time is analogous to that in Lemma \ref{fv:lem:steiner-function-bound}, and is therefore omitted.
\end{proof}

Then, we get the following result using similar arguments as in the previous sections.
\begin{theorem} \label{fv:theorem:coct}
	There exists a $\tndd$ time, polynomial space, randomized algorithm to solve \textsc{Connected Odd Cycle Transversal} in the intersection graphs of similarly sized fat objects in $\real^d$.
\end{theorem}

Note that we can reduce the standard \textsc{Odd Cycle Transversal} to \textsc{Connected Odd Cycle Tranversal}, by adding a universal vertex connected to all the vertices in the original graph. Suppose we want to solve \textsc{Odd Cycle Transversal} for a given $(G, d, \k, \Gp)$. Note that the graph $G'$ obtained by adding a universal vertex does not necessarily belong to the class of intersection graph of similarly sized fat objects. However, we observe that we can add a node corresponding to the universal vertex at the root of the weighted tree decomposition $(F, \varphi)$ of $(G, d, \k, \Gp)$. The weighted treedepth increases by at most one. Finally, we observe that, even though the degree of the universal vertex is $n$, there are at most $4$ recursive calls made from the root. Therefore, we can solve \textsc{Odd Cycle Transversal} problem on $(G, d, \P, \Gp)$ using the algorithm from above.

\begin{theorem}
	There exists a $\tndd$ time, polynomial space, randomized algorithm to solve \textsc{Odd Cycle Transversal} in the intersection graphs of similarly sized fat objects in $\real^d$.
\end{theorem}
\fi

\iflong
\section{Cycle Cover} \label{fv:sec:cyclecover}

Let $C$ be a cycle in a graph $G = (V, E)$. We use $V(C), E(C)$ to denote the set of vertices and edges in $C$ respectively. If $\C$ is a collection of cycles, then we use the notation $V(\C) = \bigcup_{C \in \C} V(C)$, and $E(\C) = \bigcup_{C \in \C} E(C)$. Finally, if $V(\C) = V(G)$, and the cycles in $\C$ are vertex disjoint, then we say that $\C$ is a \emph{cycle cover} of $G$. 

\begin{definition}[\textsc{Cycle Cover}] \label{fv:defn:cycle-cover}
		\ \\Input: An undirected graph $G = (V, E)$, and an integer $k$. 
		\\Question: Does there exist a cycle cover of $G$ of size at most $k$?
\end{definition}
Note that the case of $k = 1$ corresponds to determining whether $G$ has a Hamiltonian cycle, that is, to the \textsc{Hamiltonian Cycle} problem.

\subsection{Structural Properties}

In this section, we need a strengthened version of Lemma \ref{fv:lem:deberg-kpart}, which is possible when we have access to geometric representation of a geometric intersection graph. 

\begin{lemma} \label{fv:lem:deberg-1part}
	Let $d \ge 2$ be a constant. Then, there exists a constant $\Delta$, such that for any intersection graph $G = (V, E)$ of $n$ similarly-sized fat objects in $\real^d$ along with the geometric representation of the objects, a $1$-partition $\P$ for which $G_\P$ has maximum degree $\Delta$ can be computed in polynomial time.
\end{lemma}

In the following, we fix a $1$-partition $\P= \{V_1, V_2, \ldots, V_t\}$ such that $G_{\P}$ has maximum degree $\Delta = O(1)$. Recall that the definition of a $1$-partition implies that any $V_i \in \P$ is an induced clique in $G$. For any $V_i \in \P$, let $N(V_i)$ denote the set of neighbors of $V_i$ in the graph $G_\P$. Note that $|N(V_i)| \le \Delta = O(1)$. Finally, for any disjoint vertex subsets $U, W \subseteq V(G)$, let $E(U, W)$ denote the subset of edges in $E(G)$ with one endpoint in $U$ and another in $W$. 

We note that the structural properties proved in the rest of this subsection can be thought of as an analog of the corresponding results in Chaplick et al. \cite{ChaplickFGK019}. One important distinction between their work from ours is that we have a bound $\Delta$ on the maximum degree of $\Gp$, which we use to obtain more refined bounds. We give the formal proofs for completeness, since Chaplick et al. \cite{ChaplickFGK019} only give a sketch, and a full version with all the details is not publicly available. We also note that these results can be seen as a generalization of results of Ito et al. \cite{ito2010tractability} for the special case of \textsc{Hamiltonian Cycle}. We consider \textsc{Cycle Cover}, which is a more general problem. 

\begin{claim} \label{fv:cl:canonical}
	Let $\C = \{C_1, C_2, \ldots, C_k\}$ be a set of vertex disjoint cycles, with $k \ge 1$. Then, we can obtain another set of vertex disjoint cycles $\C' = \{C'_1, C'_2, \ldots, C'_{k'}\}$ such that (a) $k' \le k$, (b) $V(\C) = V(\C')$, and with the following two properties.
	\begin{enumerate}
		\item For any $V_i \in \P$, there exists at most one cycle $C' \in \C'$ such that $V(C') \subseteq V_i$.
		\item For any $V_i, V_j \in \P$, there exists at most one cycle $C' \in \C'$ such that $E(V_i, V_j) \cap E(C)' \neq \emptyset$. Furthermore, if there exists such a cycle $C'$, then $|E(V_i, V_j) \cap C'| \le 2$.
	\end{enumerate}
\end{claim}
\begin{proof}
	First, note that if there exist at least two cycles that are completely contained a $V_i$ for some $V_i \in \P$, then these cycles can be merged so that there exists at most one such cycle per $V_i$. Therefore, the first property is easy to satisfy.
	Consider following two operations.
	\begin{enumerate}
		\item For any $V_i, V_j$, if there exists a cycle $C \in \C$ such that $|E(V_i, V_j) \cap E(C)| \ge 3$, reroute as shown in Figure \ref{fv:fig:canonical}, cases A and B. Apply this operation repeatedly as long as $|E(V_i, V_j) \cap E(C)| \ge 3$. Finally, when this operation cannot be applied, we have that $|E(V_i, V_j) \cap E(C)| \le 2$.
		\item Consider any $V_i, V_j$, and suppose there are two edges $e_1, e_2 \in E(V_i, V_j)$ such that $e_1 \in E(C_1)$ and $e_2 \in E(C_2)$. Then, we can merge the two cycles $C_1$ and $C_2$ into one cycle (see Figure \ref{fv:fig:canonical} case C).
	\end{enumerate}
	
	\begin{figure}[h]
		\centering
		\includegraphics[scale=0.7]{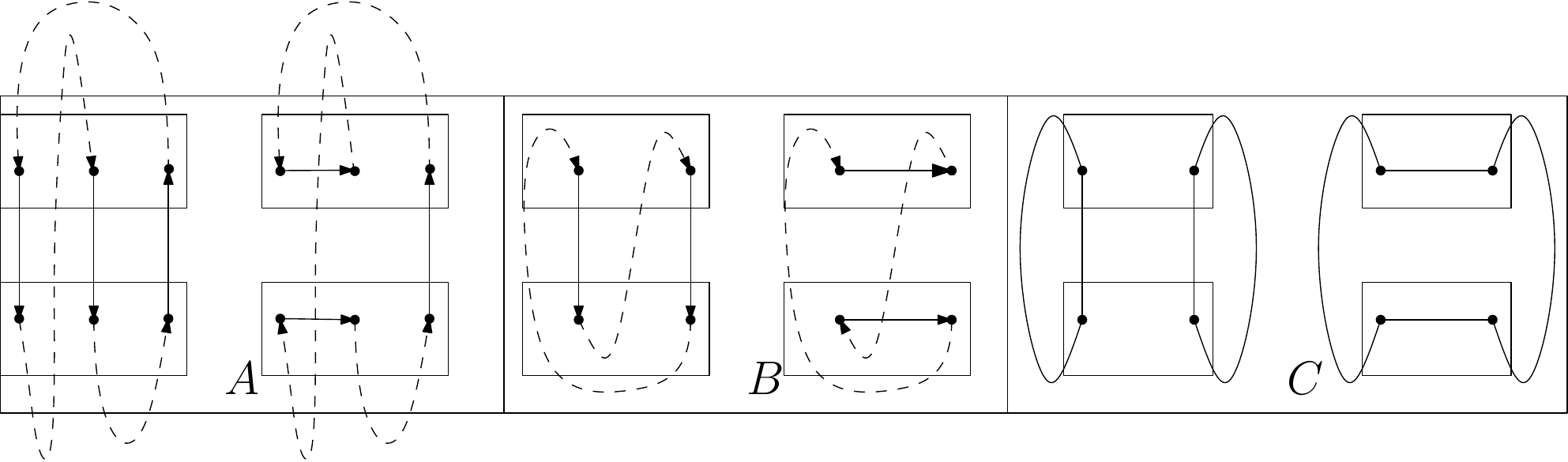}
		\caption{Different cases for rerouting in Claim \ref{fv:cl:canonical}}
		\label{fv:fig:canonical}
	\end{figure}
	
	First, observe that the set of incident vertices on the cycles does not change after applying any of the operations. Also, since we merge two cycles in the second operation, the number of cycles can only decrease.
	
	Initially, we apply the first operation repeatedly as long as it is possible. We redefine $\C' \gets \C$ after each application. Next, we apply the second operation if it is possible to do so. After each application of the second operation, we also check whether the first operation can be applied. Note that when neither of the operations can be applied to the current set of cycles $\C'$, it has the claimed properties. 
\end{proof}

If $\C$ is a set of vertex disjoint cycles satisfying the two properties from Claim \ref{fv:cl:canonical}, we say that $\C$ is a set of \emph{canonical} cycles, and a cycle $C \in \C$ is said to be a canonical cycle. In the following, when we refer to a cycle (resp.\ a set of cycles), we will assume that it is a canonical cycle (resp.\ a set of canonical cycles), unless explicitly mentioned otherwise. If $\C$ is a set of canonical cycles that is also a cycle cover of $G$, then we say that $\C$ is a \emph{canonical cycle cover} of $G$. 

We say that $u \in V_i$ is a \emph{boundary vertex} with respect to a set of cycles $\C$ if there exists a cycle $C \in \C$ such that there is an edge $uv \in E(C)$, where $v \in V_j$, $j \neq i$. We denote the set of boundary vertices in $V_i$ w.r.t.\ a set of cycles $\C$ by $B_i(\C)$. Finally, if $\C = \{C\}$, then we slightly abuse the notation and write $B_i(C)$ for $B_i(\{C\})$.

We have the following simple observation.

\begin{observation} \label{fv:obs:boundary-bound}
	For any $V_i \in \P$, and any set of cycles $\C$, $|B_i(\C)| \le 2 \min\{\Delta, |\C|\}$.
\end{observation}
\begin{proof}
	Since $\C$ is a set of canonical cycles, for every $V_i, V_j$, there exists at most one cycle $C \in \C$ such that $E(C) \cap E(V_i, V_j) \neq \emptyset$, and furthermore for such a cycle $|E(C) \cap E(V_i, V_j)| \le 2$. Now, the observation follows from the fact that the degree of $V_i$ in $G_\P$ is at most $\Delta$, which implies that there are at most $\Delta$ distinct $V_j \in \P$ such that $E(V_i, V_j) \neq \emptyset$.
\end{proof}

Before we state the following structural lemma, we define some notation.

For any $V_i \in \P$, and $V_j \in N(V_i)$, let $V_i(j) \coloneqq \{v \in V_i: N(v) \cap V_j \neq \emptyset\}$. Now, consider the bipartite graph $H_{ij} = (V_i(j) \cup V_j(i), E(V_i, V_j))$, and let $M_{ij}$ be a maximal matching in $H_{ij}$, and let $M'_{ij} \subseteq M_{ij}$ be an arbitrary subset of matching of size $\min\LR{6\Delta, |M_{ij}|}$. Furthermore, let $L_i(j)$ (resp.\ $L_j(i)$) be the endpoints of matching edges of $M'_{ij}$ in $V_i$ (resp.\ $V_j$). 

We initialize subsets $U_{i}(j), U_j(i) \gets \emptyset$, and proceed as follows.

\begin{enumerate}
	\item $|M_{ij}| \ge 6\Delta$. Note that in this case $|L_i(j)| = |L_j(i)| = |M'_{ij}| = 6\Delta$. Let $U_i(j) \gets L_i(j)$, and $U_j(i) \gets L_j(i)$.
	\item $|M_{ij}| < 6\Delta$. We first let $U_i(j) \gets L_i(j)$, and $U_j(i) \gets L_j(i)$. Then, for each $u \in L_i(j)$, we add $\min\{|N(u) \cap V_j|, 6\Delta-1\}$ neighbors of $u$ to $U_j(i)$, and analogously for every $v \in L_j(i)$ we add $\min\{|N(v) \cap V_i|, 6\Delta-1\}$ neighbors of $v$ to $U_i(j)$.
\end{enumerate}

We define $U_i \coloneqq \bigcup_{V_j \in N(V_i)} U_i(j)$. Finally, we add $\min\LR{4\Delta + 3, |V_i \setminus U_i|}$ arbitrary vertices to $U_i$. This ensures that either $|U_i| \ge 4\Delta + 3$, or $U_i = V_i$. In either case, we have that $|U_i| = O(\Delta^3)$. We have the following structural lemma regarding any canonical cycle cover of $G$.

\begin{lemma} \label{fv:lemma:canonical-boundary}
	Let $\C = \LR{C_1, C_2, \ldots, C_k}$ be a canonical cycle cover of $G$. Then we can obtain another canonical cycle cover $\C' = \LR{C'_1, C'_2, \ldots, C'_{k'}}$ of $G$, where $k' \le k$, that satisfies the following property. For every $V_i \in \P$, the set of boundary vertices w.r.t. $\C$ is a subset of $U_i$.
\end{lemma}
\begin{proof}
	For every $V_i \in \P$. Initially, all vertices in $U_i$ are ``unmarked''.
	First, we mark all the boundary vertices w.r.t.\ the set of cycles $\C$. Note that by Observation \ref{fv:obs:boundary-bound}, the number of marked vertices is at most $2\Delta$ at this point.
	
	Now we will iterate over pairs $V_i, V_j \in \P$ in an arbitrary order such that $V_i$ and $V_j$ are neighbors in $G_\P$. Note that, since $\C$ is a set of canonical cycle, there is at most one cycle $C \in \C$ such that $0 < |E(C) \cap E(V_i, V_j)| \le 2$.
	Note that $1 \le |B_i(C)|, |B_j(C)| \le 2$ using Observation \ref{fv:obs:boundary-bound}. We will show how to modify the cycles in $\C$ to obtain another set of cycles $\C'$ such that the following properties are satisfied.
	\begin{enumerate}
		\item $\C'$ is a set of vertex disjoint canonical cycles with $V(\C') = V(G)$, $|\C'| \le |\C|$, and
		\item If there exists a cycle $C' \in \C'$ such that $E(C') \cap E(V_i, V_j) \neq \emptyset$, then $B_i(C) \subseteq U_i$, and $B_j(C) \subseteq U_j$.
	\end{enumerate}
	Furthermore, during the rerouting process, we will mark at most $2$ new vertices of $U_i$ (resp.\ $U_j$). Finally, at the end of the iteration we will set $\C \gets \C'$, and proceed to the next iteration. 
	
	Suppose at the beginning of the iteration there is no $C \in \C$ such that $E(C) \cap E(V_i, V_j) \neq \emptyset$, or if for the cycle $C \in \C$ with $E(C) \cap E(V_i, V_j) \neq \emptyset$, we already have that $B_i(C) \subseteq U_i$ and $B_j(C) \subseteq U_j$, then the set $\C$ satisfies the required properties. Therefore in this case we proceed to the next iteration without modifying the set $\C$. Now let us discuss the interesting case when we have to modify $\C$.
	
	Without loss of generality, assume that $B_i(C) \not\subseteq U_i$, which also implies that $B_i(C) \not\subseteq U_i(j)$. We consider two cases.
	
	\textbf{Case 1.} $6\Delta = |M'_{ij}| < |M_{ij}|$. Then, we have that $|U_i(j)| = |U_j(i)| = |M'_{ij}| = 6\Delta$. Note that initially at most $2\Delta$ vertices in $U_i$ (resp.\ $U_j$) were marked, and in each of the previous iterations we have marked at most $2$ new vertices. Furthermore, before the iteration corresponding to $\{V_i, V_j\}$ have been at most $2(\Delta-1)$ iterations involving either $V_i$ or $V_j$. Therefore, at the beginning of this iteration, at most $2\Delta + 2(\Delta-1) = 4\Delta-2$ vertices in $U_i$ (resp.\ $U_j$) are marked. However, since $|M'_{ij}| = 6\Delta$, there must be at least two edges $e_1 = u_iu_j, e_2 = v_iv_j \in M'_{ij}$, such that $u_i, v_i \in U_i$ are unmarked, and $u_j, v_j \in U_j$ are unmarked.
	
	Consider $u_i$. Since $u_i \in U_i$ was unmarked, it is not a boundary vertex w.r.t.\ any cycle in $\C$. Therefore, the cycle $C(v_i) \in \C$ incident on $u_i$ is of the form $(\ldots, a_i, u_i, b_i, \ldots)$, where $a_i, b_i \in V_i$. Since $V_i$ is a clique, $a_ib_i \in E(G)$, therefore we can short-cut $C(v_i)$ to obtain $C'(v_i)$ which is of the form $(\ldots, a_i, b_i \ldots)$ \footnote{If $C(v_i)$ contains only three vertices $\{a_i, u_i, b_i\}$ then we cannot do such a modification. However, in this case, when can can reroute $C$ as $(\ldots, u_i, a_i, b_i, v'_i \ldots)$. Note that the number of cycles decreases in this case. Other corner cases can be handled in a similar manner, which we do not discuss here.}. Note that it is possible to have $C(v_i) = C$, in which case we perform simultaneous replacements to $C$ at two different places. Analogous claims also hold for $v_i$, as well $u_j, v_j$. See figure \ref{fv:fig:cycle-bypass} for an illustration.
	
	\textbf{Case 2.} $|M'_{ij}| = |M_{ij}| < 6\Delta$. Note that in this case, $|L_i(j)| = |L_j(i)| = |M_{ij}|$. Let $u \in V_i$ be a vertex such that $u \in B_i(C) \setminus U_i(j)$, and let $v \in N(u) \cap V_j$ such that $uv \in E(C) \cap E(V_i, V_j)$. Note that $v$ must be in $U_j(i)$ -- otherwise $M_{ij} \cup \{uv\}$ is a matching of size at most $6\Delta$, contradicting the maximality of $M_{ij}$. Now, we add $\min\{6\Delta, |N(v) \cap V_i|\}$ neighbors of $v$ from $V_i \setminus L_i(j)$ to $U_i(j)$. Since $v \not\in U_i(j)$, it must be the case that we have added exactly $6\Delta$ neighbors of $v$ to $U_i(j)$. Using an argument similar to the previous paragraph, there must exist an unmarked neighbor $u_i \in N(u) \cap V_i$. We replace the edge $uv \in E(C)$ with the edge $uu_i$ to obtain a new cycle $C'$. Note that we also have to perform replacement in the cycle $C(u_i)$ as in the previous paragraph. The final set of cycles $\C'$ obtained satisfies the desired properties. 
	\begin{figure}
		\centering
		\includegraphics[scale=0.8]{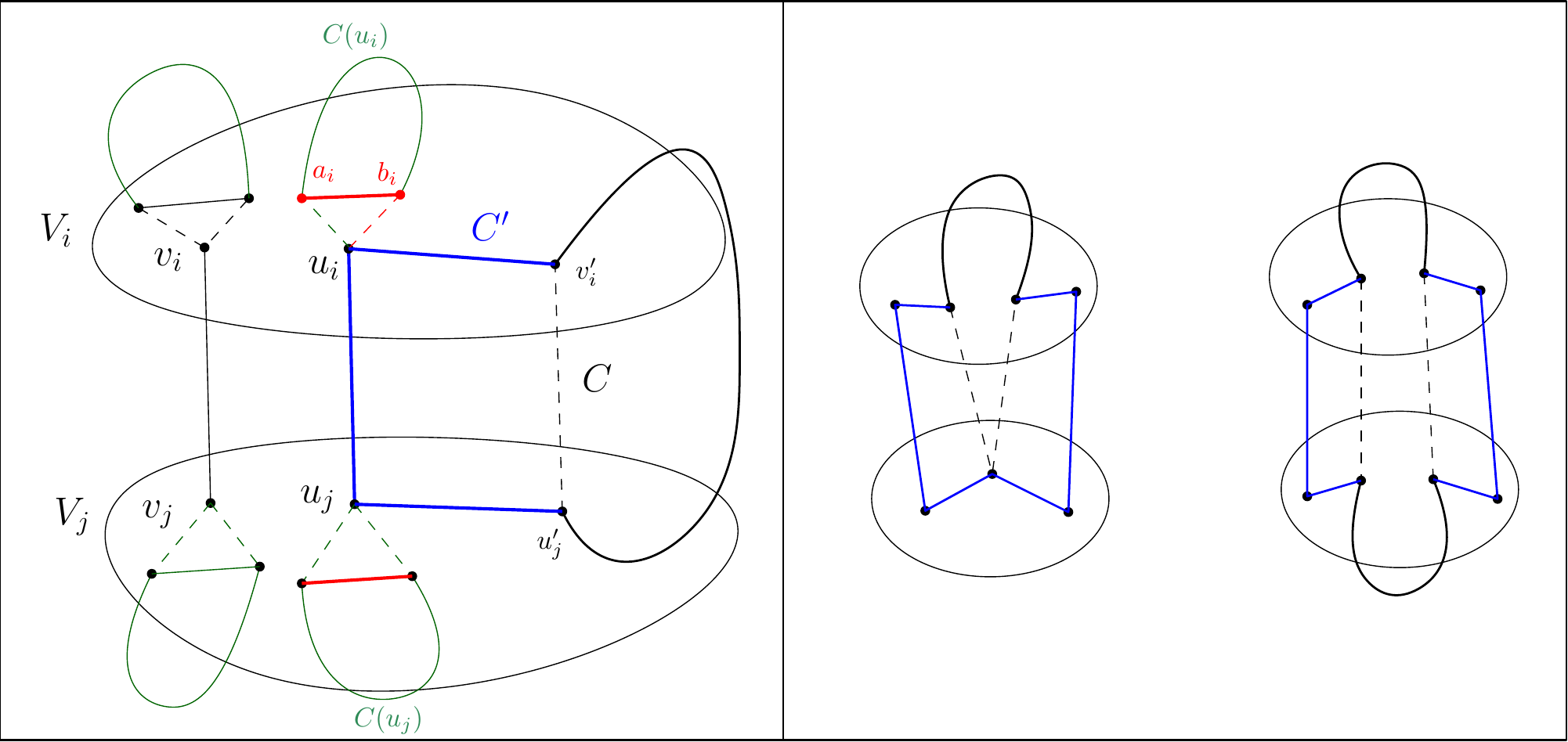}
		\caption{An illustration for modifying cycles. 
		\\Left: Suppose $u'_i \not\in U_i$ and $u'_j \not\in U_j$. Then, we reroute $C$ using the thick blue edges to obtain another cycle $C'$. Furthermore, we also have to reroute cycles $C(u_i)$ and $C(u_j)$ using the thick red edges. Note that $v_i, v_j$ can also be used in a similar manner to reroute another ``problematic'' edge in $E(C) \cap E(V_i, V_j)$ (see the figure on the right).
		\\Right: Some other cases for rerouting. Dashed edges in the cycle $C$ are replaced with the blue edges to obtain new cycle $C'$. Rerouting of the cycles $C(w)$ for the new vertices $w$ incident on the blue edges is not shown.}
		\label{fv:fig:cycle-bypass}
	\end{figure}
	
	Therefore, instead of the boundary vertices $B_i(C)$, we can reroute the cycle $C$ such that it uses the desired number of boundary vertices from the set $\{u_i, v_i, u_j, v_j\}$. Let $C'$ denote the cycle obtained after the modification. We also short-cut cycle(s) $C(w)$ for $w \in B_i(C') \cup B_j(C')$ as discussed in the previous paragraph to obtain $C'(w)$. Let $\C'$ be the set of cycles obtained by replacing in $\C$ any cycle $C_\ell \in \C$ involved in a modification, by its modified version $C'_\ell$. Finally, we mark the vertices in $B_i(C')$ and in $B_j(C')$. Observe that $\C'$ satisfies the properties claimed earlier, and we mark at most two vertices in $U_i$ (resp.\ $U_j$). This completes the proof by induction.
\end{proof}

\fi

\iflong
Let $U = \bigcup_{V_i \in \P} U_i$, and let $H = G[U]$. Furthermore, let $\P' = \{U_i: V_i \in \P\}$ be a $1$-partition of $H$, and let $H_{\P'}$ be the graph isomorphic to $\Gp$, where a vertex $V_i$ is replaced by corresponding $U_i$.

\begin{observation} \label{fv:obs:canonical-sets}
	For any $V_i \in \P$, we have that $U_i \subseteq V_i$ such that $|U_i| = O(\Delta^3)$. Furthermore, the sets $\LR{U_i}_{V_i \in \P}$, and $(H, d, \P', H_{\P'})$ as defined above, can be computed in polynomial time.
\end{observation}

\begin{theorem}
	$G$ has a canonical cycle cover of size $k$ iff $H$ has a canonical cycle cover of size $k$.
\end{theorem}
\begin{proof}
	First we prove the forward direction. Let $\C = \LR{C_1, C_2, \ldots, C_{k}}$ with $V(\C) = V(G)$ be a canonical cycle cover of $G$, such that for every $V_i \in \P$, the set of boundary vertices w.r.t. $\C$ is a subset of $U_i$, as guaranteed by Lemma \ref{fv:lemma:canonical-boundary}. Now consider any $V_i \in \P$ with $U_i \subsetneq V_i$, and a subset $\C_i \coloneqq \{C_1, C_2, \ldots, C_t\} \subseteq \C$, that contain at least one vertex vertex from $V_i$. Note that $|\C_i| \le \Delta$, and $|B_i(\C_i)| \le 2\Delta$. Furthermore, since $U_i \subsetneq V_i$, we have that $|U_i| \ge 4\Delta+3$. Initially, all vertices in $U_i$ are unmarked. Then, we mark all vertices in $B_i(\C_i) \subseteq U_i$.
	
	If there is a cycle $C \in \C_i$ such that $V(C) \subseteq V_i$ (Type 1), then we obtain another cycle $C'$ that spans three unmarked vertices from $U_i$. We mark these vertices. Note that by Claim \ref{fv:cl:canonical}, there exists at most one cycle of Type 1, since we assume $\C$ is a canonical cycle cover of $G$.
	
	Now consider a cycle $C \in \C_i$ such that $V(C) \subsetneq V_i$ (Type 2), and let us focus on a minimal sub-path $\pi$ of $C$ of the form $(\ldots b_1, q_1, q_2, \ldots, q_\ell, b_2, \ldots)$ that lies completely inside $V_i$, such that $b_1, b_2 \in B_i(C)$, and $\{q_1, q_2, \ldots, q_\ell\} \cap B_i(C) = \emptyset$. That is, the cycle $C$ enters $V_i$ via a boundary vertex $b_1$, visits internal vertices $q_1, \ldots, q_\ell$, and exits $V_i$ via another boundary vertex $b_2$. We obtain another cycle $C'$ by replacing the sub-path $\pi$ with $(\ldots, b_1, s_1, s_2, b_2,\ldots)$, where $s_1, s_2$ are arbitrary unmarked vertices in $U_i$. We mark these vertices. We perform this rerouting for all cycles $C \in \C_i$, and all minimal sub-paths of the form described. Note that $|B_i(\C_i)| \le 2\Delta$ (Observation \ref{fv:obs:boundary-bound}), each boundary vertex can be involved in at most one rerouting, and we use exactly two unmarked vertices per rerouting. Therefore, at most $4\Delta$ unmarked vertices are sufficient for Type 2. Since $|U_i| \ge 4\Delta + 3$, it is easy to see that we have enough unmarked vertices for all reroutings of Type 1 and 2.
	 
	Finally, if there are some unmarked vertices in $U_i$ at the end, we arbitrarily select one edge $uv \in E(C)$ for some rerouted cycle $C \in \C_i$, where $u, v \in U_i$. We make a final rerouting of this cycle by adding all unmarked vertices in $U_i$ between $u$ and $v$ in an arbitrary order. Note that after making this replacement for every $V_i \in \P$, we obtain a canonical cycle cover for $H$ of the same size.
	
	Now let us look at the reverse direction. Let $\C = \LR{C_1, C_2, \ldots, C_{k}}$ be a canonical cycle cover of $H$. Consider a $W_i \in \P'$ such that $W_i \subsetneq V_i$, and let $C \in \C$ be a cycle such that $E(C)$ uses at least one edge $uv$ from $H[W_i]$, i.e., $u, v \in W_i$. Note that such a cycle $C$, and the corresponding edge $uv$ must exist, since $U_i \subsetneq V_i$, which implies that $|U_i| \ge 4\Delta+3$, whereas the number of boundary vertices in $U_i$ is at most $2\Delta$ by Observation \ref{fv:obs:boundary-bound}. We simply insert all vertices in $V_i \setminus U_i$ in an arbitrary order between $u$ and $v$ in the cycle $C$. After performing these replacements for every $U_i \in \P'$, it is easy to see that we obtain a canonical cycle cover of $G$ of the same size.
\end{proof}

\subsection{Algorithm}

First, we prove the following key lemma. 

\begin{lemma} \label{fv:lem:unweighted-treedepth}
	The (unweighted) treedepth of $H$ is $\ndd$, and a treedepth decomposition $(F, \varphi)$ of treedepth $\ndd$ can be computed in time $\tndd$ and polynomial space.
\end{lemma}
\begin{proof}
	Given $(G, d, \P, \Gp)$, we use Theorem \ref{fv:thm:treedepth} to compute weighted treedepth decomposition $(F', \varphi')$ of weighted treedepth $\ndd$, in time $\tndd$, and using polynomial space. 
	
	We now discuss how to obtain the (unweighted) treedepth decomposition $(F, \varphi)$ of the claimed depth. Recall that node $u_i \in V(F')$, we have that $\varphi'(u_i) = V_i$ for some $V_i \in \P$. Correspondingly, let $\pi_i$ be a path of length $|W_i| = O(\Delta^3)$ to be added to the forest $F$.
	
	Let $u_j \in V(F')$ be the parent of $u_i \in V(F')$ (except for the root). We add an edge between the last node of the path $\pi_j$ (corresponding to $u_j, W_j$) to the first vertex of path $\pi_i$ (corresponding to $u_i, W_i$), such that the latter is a child of the former. In other words, if we contract each path $\pi_i$ in $F$ into a single node, we obtain a graph that is isomorphic to $F'$. Finally, we obtain the bijection $\varphi: V(F) \to V(H)$ by arbitrarily mapping each vertex on $\pi_i$ to a unique vertex from $W_i$.
	
	For analyzing the unweighted treedepth of $H$, consider any root-leaf path $P$ in $F$. From the last paragraph, this corresponds to a root-leaf path $P'$ in $F'$, which we obtain by contracting each sub-path $\pi_i$ of length $O(\Delta^3)$ into a single vertex. We have the following bound on the number of vertices on $P$.
	$$ |P| \le \Oh(\Delta^3) \cdot |P'| \le \Oh(\Delta^3) \cdot \sum_{u_i \in P'} \omega(u_i) = \Oh(\Delta^3) \cdot \omega(P') = \ndd $$
	Where the second inequality follows from the fact that $\omega(u_i) = \log(1+|V_i|) \ge 1$, and the last inequality follows from the bound on the $\wtd(G)$, and since $\Delta$ is a constant. It follows that the (unweighted) treedepth of $(F, \varphi)$ thus constructed is $\ndd$. Finally, observe that after computing $(F', \varphi')$, we can construct $(F, \varphi)$ in polynomial time and polynomial space.
\end{proof}

Finally, we appeal to the following result from \cite{NederlofPSW20}.

\begin{theorem}[\cite{NederlofPSW20}]
	Given a graph $H$, and its treedepth decomposition of (unweighted) treedepth $d$, there exists a $2^{O(d)} \cdot n^{O(1)}$ time, polynomial space randomized algorithm, to solve \textsc{Cycle Cover}.
\end{theorem}

Note that the algorithm in Nederlof et al.\ \cite{NederlofPSW20} is actually for the \textsc{Partial Cycle Cover} problem, which asks to determine whether there exists a set of at most $k$ vertex disjoint cycles that span \emph{exactly} $\ell$ vertices. Note that the \textsc{Cycle Cover} problem corresponds to the special case where $\ell = n$. Finally, given an instance of \textsc{Hamiltonian Path}, where the task is to find a path spanning the set of vertices, we can reduce it to \textsc{Hamiltonian Cycle} by adding a universal vertex. Thus, we have the following result. 

\begin{theorem}
	There exists a $\tndd$ time, polynomial space, randomized algorithm, to solve \textsc{Cycle Cover} in the intersection graphs of similarly sized fat objects in $\real^d$. In particular, this implies analogous results for \textsc{Hamiltonian Cycle} and \textsc{Hamiltonian Path} in the intersection graphs of similarly sized fat objects in $\real^d$.
\end{theorem}
\fi

\section{Conclusion and Open Questions}
In this paper, following \cite{BergBKMZ20}, we consider various graph problems in the intersection graphs of similarly sized fat objects. Our running times for \textsc{Independent Set, $r$-Dominating Set, Steiner Tree, Connected Vertex Cover, Feedback Vertex Cover, (Connected) Odd Cycle Transversal, Hamiltonian Cycle} are of the form $\tndd$---matching that in \cite{BergBKMZ20}---but we improve the space requirement to be polynomial. Due to some technical reasons, we are not able to achieve a similar result for \textsc{Connected Dominating Set} which is also considered by \cite{BergBKMZ20}. We leave this as an open problem.

Kisfaludi-Bak \cite{kisfaludi2020hyperbolic} used some of the ideas from \cite{BergBKMZ20} in the context of (noisy) unit ball graphs in $d$-dimensional hyperbolic space. In particular, he gave subexponential and quasi-polynomial time (and space) algorithms for problems such as \textsc{Independent Set, Steiner Tree, Hamiltonian Cycle} using a notion similar to the weighted treedepth. Using our techniques, it should be possible to improve the space requirement of these algorithms to polynomial, while keeping the running time same (up to possibly a multiplicative $\Oh(\log n)$ factor in the exponent in some cases). We leave the details for a future version.

Finally, our algorithms for the connectivity problems such as \textsc{Steiner Tree, Connected Vertex Cover, (Connected) Odd Cycle Transversal}, and that for \textsc{Cycle Cover} use an adapted version of the \cnc technique (\cite{NederlofPSW20,HegerfeldK20,cygan2011solving}). \cnc technique crucially uses the Isolation Lemma (cf. Lemma \ref{fv:lem:isolation}), and hence these algorithms are inherently randomized. We note that recently there has been some progress toward derandomizing Cut\&Count \cite{NederlofPSW21} for problems such as \textsc{Hamiltonian Cycle} on graphs of bounded treedepth. This may also have some consequences for our algorithms.
\bibliography{references}
 
\appendix

\end{document}